\documentclass[sigconf,nonacm]{acmart}

\newif\ifconf
\newif\iftr
\newif\ifnblnd

\trfalse
\conftrue
\nblndfalse

\trtrue
\conffalse
\nblndtrue

\iftr
\setcopyright{none}
\renewcommand\footnotetextcopyrightpermission[1]{}
\fi

\ifconf
\copyrightyear{2025}
\acmYear{2025}
\setcopyright{acmlicensed}\acmConference[SC '25]{The International Conference for High Performance Computing, Networking, Storage and Analysis}{November 16--21, 2025}{St. Louis, MO, USA}
\acmBooktitle{The International Conference for High Performance Computing, Networking, Storage and Analysis (SC '25), November 16--21, 2025, St. Louis, MO, USA}
\fi

\usepackage{booktabs} 

\newif\ifsq     

\newif\ifsqCAP
\newif\ifsqVS
\newif\ifsqEN
\newif\ifsqTIT

\sqfalse
\sqCAPfalse
\sqENfalse
\sqVSfalse
\sqTITfalse

\ifconf
\sqtrue
\sqCAPtrue
\sqENtrue
\sqVStrue
\sqTITtrue
\fi

\usepackage[font={sf,small}]{caption}
\usepackage[font={sf,small}]{subcaption}

\newcommand{\vspaceSQ}[1]{\ifsqVS\vspace{#1}\fi}
\newcommand{\enlargeSQ}[1]{\ifsqEN\enlargethispage{\baselineskip}\fi}

\ifsqTIT
\usepackage[compact]{titlesec}
\titlespacing*{\section}{0pt}{6pt}{3pt}
\titlespacing*{\subsection}{0pt}{4pt}{2pt}
\titlespacing*{\subsubsection}{0pt}{2pt}{3pt}
\fi

\usepackage[sc]{mathpazo}

\usepackage{breqn}

\usepackage{graphicx}
\usepackage{booktabs}
\usepackage{epstopdf}
\usepackage{url}
\usepackage{placeins}
\usepackage{multirow}
\usepackage{rotating}
\usepackage{pbox}
\usepackage[normalem]{ulem}

\usepackage{tabularx, makecell}

\usepackage{listings}

\usepackage{cleveref}
\usepackage[utf8]{inputenc}
\crefname{section}{§}{§§}
\Crefname{section}{§}{§§}

\usepackage[10pt]{moresize}

\usepackage{xcolor}
\definecolor{darkgrey}{RGB}{70,70,70}
\definecolor{lightgrey}{RGB}{200,200,200}

\usepackage{enumitem}

\usepackage[linesnumbered,ruled,vlined]{algorithm2e}
\SetAlFnt{\scriptsize\tt}
\SetAlCapFnt{\scriptsize\sf}
\SetAlCapNameFnt{\scriptsize\sf}

\newcommand{\goal}[1]{\noindent\textcolor{red}{[Goal: #1]}\par}

\newcommand{\macb}[1]{\textbf{\textsf{#1}}}

\usepackage{bm}

\usepackage{scalerel,stackengine}
\stackMath
\newcommand\rwh[1]{%
\savestack{\tmpbox}{\stretchto{%
  \scaleto{%
      \scalerel*[\widthof{\ensuremath{#1}}]{\kern-.6pt\bigwedge\kern-.6pt}%
          {\rule[-\textheight/2]{1ex}{\textheight}}
            }{\textheight}%
}{0.5ex}}%
\stackon[1pt]{#1}{\tmpbox}%
}

\def\HiLiGA{\leavevmode\rlap{\hbox to \hsize{\color{black!10}\leaders\hrule height 1\baselineskip depth 1ex\hfill}}}
\def\HiLiGB{\leavevmode\rlap{\hbox to \hsize{\color{black!25}\leaders\hrule height 1\baselineskip depth 1ex\hfill}}}
\def\HiLiGC{\leavevmode\rlap{\hbox to \hsize{\color{black!40}\leaders\hrule height 1\baselineskip depth 1ex\hfill}}}
\def\HiLiGD{\leavevmode\rlap{\hbox to \hsize{\color{black!55}\leaders\hrule height 1\baselineskip depth 1ex\hfill}}}
\def\HiLiGE{\leavevmode\rlap{\hbox to \hsize{\color{black!70}\leaders\hrule height 1\baselineskip depth 1ex\hfill}}}
\def\HiLiGF{\leavevmode\rlap{\hbox to \hsize{\color{black!85}\leaders\hrule height 1\baselineskip depth 1ex\hfill}}}

\usepackage{algpseudocode}
\usepackage{tikz}
\usetikzlibrary{calc}
\usepackage{xcolor}
\makeatletter

\renewcommand{\goal}[1]{}

\usepackage{makecell}

\definecolor{vlgray}{rgb}{0.77 0.77 0.77}
\definecolor{ablack}{rgb}{0.2 0.2 0.2}

\usepackage{tikz}
\usetikzlibrary{tikzmark}

\usepackage{xpatch}
\expandafter\xpatchcmd
\csname pgfk@/tikz/every picture/.@cmd\endcsname
{\thepage}{\arabic{page}}{}{}

\tikzstyle{comment} = [draw, fill=blue!70, text=white, text width=3cm, minimum height=1cm, rounded corners, align=left, font=\scriptsize]
\tikzstyle{background_alg} = [draw, fill=blue!20, opacity=0.4, inner sep=4pt, rounded corners=2pt]

\usepackage{fontawesome}
\usepackage{pifont}
\usepackage{textcomp}

\newfloat{lstfloat}{htbp}{lop}
\floatname{lstfloat}{Listing}

\begin{document}

\title{Higher-Order Graph Databases}
\renewcommand{\shorttitle}{Higher-Order Graph Databases}


\ifconf
\else
\author{Maciej Besta$^{1*\dagger}$,
Shriram Chandran$^{1* \dagger}$,
Jakub Cudak$^{2}$,
Patrick Iff$^{1}$,
Marcin Copik$^{1}$,\\
Robert Gerstenberger$^{1}$,
Tomasz Szydlo$^{2,3}$,
Jürgen Müller$^{4}$,
Torsten Hoefler$^{1}$}
\affiliation{\vspace{0.3em}$^1$ETH Zurich;
$^2$AGH-UST;
$^3$NCL UK;
$^4$BASF SE \\
\country{{$^*$}Corresponding authors}
\country{{$^\dagger$}Alphabetical order}}

\renewcommand{\shortauthors}{M. Besta et al.}
\fi

\begin{abstract}
Recent advances in graph databases (GDBs) have been driving interest in large-scale analytics, yet current systems fail to support higher-order (HO) interactions beyond first-order (one-hop) relations, which are crucial for tasks such as subgraph counting, polyadic modeling, and HO graph learning. We address this by introducing a new class of systems, higher-order graph databases (HO-GDBs) that use lifting and lowering paradigms to seamlessly extend traditional GDBs with HO. We provide a theoretical analysis of OLTP and OLAP queries, ensuring correctness, scalability, and ACID compliance. We implement a lightweight, modular, and parallelizable HO-GDB prototype that offers native support for hypergraphs, node-tuples, subgraphs, and other HO structures under a unified API. The prototype scales to large HO OLTP \& OLAP workloads and shows how HO improves analytical tasks, for example enhancing accuracy of graph neural networks within a GDB by 44\%. Our work ensures low latency and high query throughput, and generalizes both ACID-compliant and eventually consistent systems.
\end{abstract}

\iftr
\fi

\maketitle

{\noindent\macb{Code:} \url{https://github.com/spcl/HO-GDB}}

\section{INTRODUCTION}
\label{sec:intro}

Graph databases (GDBs)~\cite{besta2023demystifying} are a specialized class of data management systems designed for storing, querying, and analyzing datasets structured as graphs. Unlike relational or document-oriented databases, GDBs natively support complex relationships between entities, enabling expressive queries over interconnected data. Modern systems such as Neo4j~\cite{neo4j_book}, TigerGraph~\cite{tiger_graph_links}, and JanusGraph~\cite{janusgraph} use rich data models like the Labeled Property Graph (LPG)~\cite{angles2017foundations} or Resource Description Framework (RDF)~\cite{rdf_links} to capture semantic information in graph-structured data through node and edge labels, properties, and types. GDBs have seen widespread adoption in domains such as computational biology, recommendation systems, and fraud detection, where data is inherently relational and large-scale. 

\begin{figure*}[th]
    \centering
    \vspaceSQ{-0.5em}
    \includegraphics[width=1.0\textwidth]{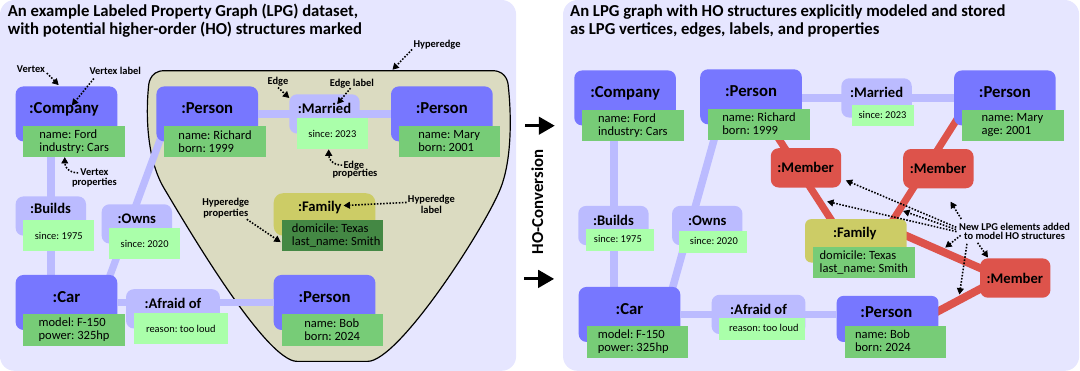}
    \vspace{-2em}
    \caption{\textmd{\textbf{Left side:} An example Labeled Property Graph (LPG) with added higher-order (HO) structures. \textbf{Right side:} the same LPG with its HO structures modeled as LPG vertices, edges, labels, and properties. This enables these HO structures to be stored and used seamlessly in most graph databases.}}
    \label{fig:example}
    \vspaceSQ{-0.5em}
\end{figure*}

\begin{figure*}[h!t]
    \centering
    \vspaceSQ{-1em}
    \includegraphics[width=0.95\textwidth]{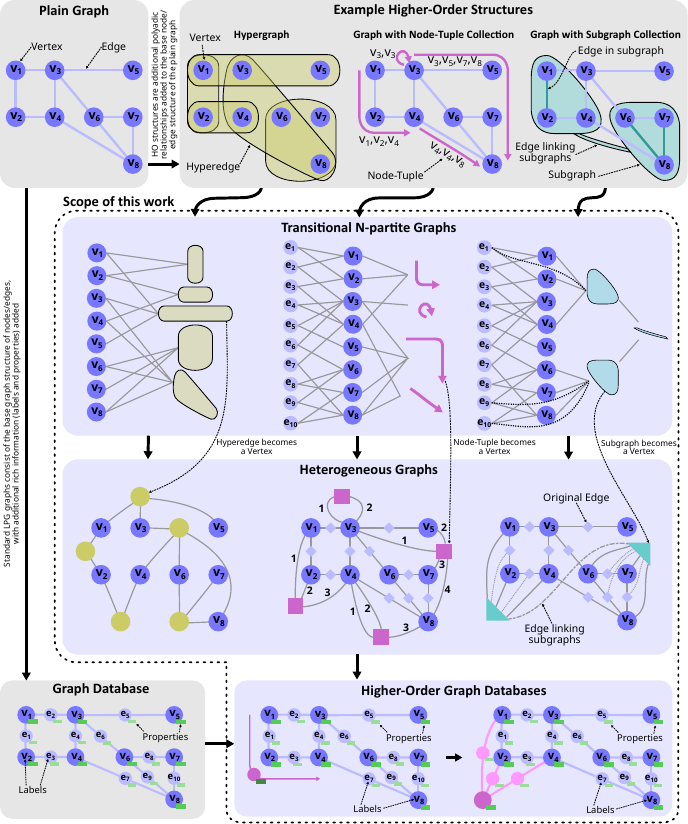}
    \vspace{-1em}
    \caption{\textmd{Overview of the introduced HO-GDB paradigm and new class of systems. \textbf{Top row:} An example plain (i.e. non-HO) graph and the same graph with example HO structures. \textbf{Top-mid row (``Transitional N-partite Graphs''):} an intermediate stage of translating the HO structures into a format digestible by any GDB. \textbf{Bottom-mid row (``Heterogeneous Graphs''):} the result of the ``lowering'' transformation that encodes HO structures as heterogeneous graphs, which are a special simplified case of the LPG data model. The \textbf{bottom row:} the combination of a standard LPG dataset (bottom-left) and the generated heterogeneous graphs that encode HO, gives the HO-enriched LPG dataset that forms the basis for higher-order graph databases (HO-GDBs). The colors of nodes and edges indicate their respective label with entities having the same color also sharing the same label.}}
    \label{fig:scheme}
\end{figure*}

Simultaneously, recent advancements in network science and graph representation learning have led to a surge of interest in higher-order (HO) graph analytics, which extend beyond pairwise (dyadic) relationships to model polyadic interactions~\cite{Alvarez_Rodriguez_2021}. HO graphs encompass several constructs, including hypergraphs, simplicial and cell complexes, motifs and motif-based graphs, subgraph collections, and nested graph structures~\cite{Bick_2023}. These models enable the representation of complex multi-way interactions such as co-authorship cliques, multi-drug interactions, friend groups, or nested molecular structures, which cannot be fully captured using binary edges alone. For example, simplicial complexes have been used in finite element mesh computations (e.g., fluid dynamics or structural analysis simulations)~\cite{hoffman2016fenics}, social network analyses (e.g., modeling group influence on adoption of behaviors), sensor networks (e.g., ensuring that a set of sensors covers a given area~\cite{muhammad2007dynamic}), computational biology (e.g., computing alpha complexes~\cite{benkaidali2014computing}), material science (e.g., computing persistent homology), or in vision (e.g., quantifying holes in shapes). Moreover, the ability to express such higher-order constructs as first-class citizens in a graph (i.e., attribute them with properties, and link them via edges) enables rich modeling of semantically meaningful structures. For instance, customer buying patterns or manufacturing incidents can be represented as HO constructs annotated with contextual metadata or connected to ontologies and similarity graphs. HO graph models are increasingly critical for achieving state-of-the-art performance in graph neural networks (GNNs)~\cite{xu2018powerful}, enabling expressiveness beyond the Weisfeiler–Lehman test~\cite{weisfeiler1968reduction} and mitigating issues like oversquashing or indistinguishability in message-passing architectures. Yet, this progress in analytics and learning over HO data remains largely disconnected from GDB system architectures and HPC in general.


Existing GDBs offer no native support for HO graph models, either in their core data model, transactional logic, or query execution engines. As a result, HO graph structures must be flattened or emulated within first-order schemas, losing essential semantics and increasing complexity for both developers and systems. This gap raises several questions: How can we extend graph data models in GDBs to natively support higher-order relationships? How can transactional guarantees and indexing structures be redefined to accommodate HO structures and their complex interdependencies? What programming abstractions and query languages are required to express a computation over subgraphs, simplices, or hyperedges? Finally, how can we achieve scalability for higher-order workloads that mix transactional access with higher-order analytical and learning pipelines?

We answer all the above questions by introducing a unified framework for integrating HO structures into graph databases, which lays the foundation for a new class of GDB systems: the Higher-Order Graph Database (HO-GDB) (\textbf{contribution~1}). Our central idea is to reformulate HO graph constructs, such as node tuples, hyperedges, simplices, or subgraphs, as instances of heterogeneous graphs~\cite{sun2021heterogeneoushypergraphembeddinggraph, 9354594}, which are a natural special case of Labeled Property Graphs~\cite{angles2017foundations} that allow for multiple node and edge types. Heterogeneous graphs are widely supported as first-class citizens in graph database systems, and we use this flexible model to encode different forms of higher-order interactions as specifically labeled entities within an LPG schema. This approach enables a direct and expressive mapping of higher-order structures into existing GDB engines, preserving transactional semantics and providing an interface to query. The results of our analysis enable native support for higher-order, polyadic, nested, and/or structured graph data at the data model level without requiring flattening or external processing.

\enlargeSQ

We develop HO-GDB, a lightweight, modular, parallelizable architecture built on top of Neo4j, for a higher-order graph database that ensures performance and correctness guarantees across complex HO workloads \textbf{(contribution~2}). Our architecture extends the core components of a GDB, such as the storage layer, indexing subsystem, and query execution engine, to similarly handle higher-order constructs transparently. We formalize the transactional semantics for HO data, including the isolation and atomicity of operations. To support efficient query execution and mitigate latency with the underlying GDB, we also introduce motif-level parallelism, allowing for concurrent processing of different HO structures, boosting overall system throughput. Our implementation demonstrates that HO-GDBs can scale to realistic workloads, integrate with existing standard graph ecosystems, and serve as a practical foundation for transactional and analytical operations over higher-order graph data.

We further demonstrate the expressiveness and analytical advantages of HO-GDBs by presenting a set of novel workloads enabled by HO  (\textbf{contribution~3}), such as path traversal over HO graphs and end-to-end integration with HO graph neural networks (HO GNNs) that operate natively on subgraph-based structures stored within the database. This allows for query-driven HO GNN training, efficient substructure sampling, integrated HO pattern mining, and others. 

\enlargeSQ

Finally, we empirically evaluate our HO-GDB prototype across selected workloads, demonstrating scalability and practical performance, supported by a rigorous theoretical analysis of performance and storage overhead (\textbf{contribution~4}). We benchmark our system using real-world datasets and measure latency and throughput for both online transactional (OLTP) and online analytical (OLAP) operations. Our evaluation confirms that HO-GDB can scale to complex HO graphs with hundreds of thousands of entities, and sustain low-latency query execution even under complex higher-order workloads. We demonstrate significant gains in expressiveness for graph-based learning over HO data. We conclude that HO graph databases not only generalize the current generation of graph databases but also open up new opportunities for system-level innovation, data modeling, and integrated machine learning in the graph data management landscape.




\begin{figure*}[t]
    \centering
    \vspaceSQ{-1em}
    \includegraphics[width=1.0\textwidth]{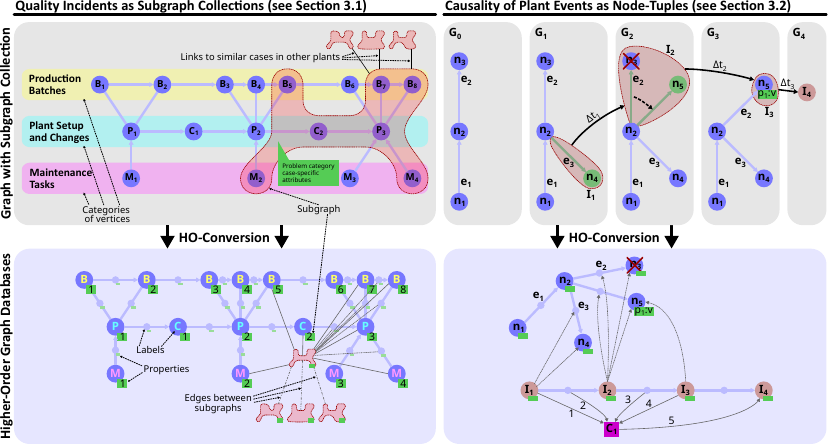}
    \vspace{-2em}
    \caption{An illustration of two higher-order (HO) practical use cases: (\textbf{left}) quality incidents as subgraph collections (details in Section~\ref{sec:use-case-incidents}), and (\textbf{right}) causality of processing plant events as node-tuples (details in Section~\ref{sec:use-case-causality}).}
    \label{fig:use-case}
\end{figure*}

\section{BACKGROUND}
\label{sec:background}

Our work combines concepts from two key domains: graph databases and higher-order graph data models.

\subsection{Graph Databases}

GDBs are specialized database systems designed to efficiently store, query, and process graph-structured data~\cite{angles2018introduction, davoudian2018survey, besta2023demystifying}. Unlike traditional relational databases, which use tables and rows to represent data, GDBs use nodes and edges to capture relationships in the data explicitly. This allows for efficient graph-based traversal operations, making them well-suited for datasets with complex relationships, such as social networks, knowledge graphs, and recommendation systems.



The established graph data model for native GDBs is the \textbf{Labeled Property Graph (LPG)}~\cite{angles2017foundations}. In LPGs, nodes model entities and can have multiple labels or types and edges model relationships between nodes and can be directed or undirected. Furthermore, both nodes and edges store information in the form of properties which are key-value pairs. 
Formally, an LPG is defined as a tuple $G = (V, E, L, l_n, l_e, K, W, p_n, p_e)$ where $V$ is the set of nodes, $E$ is the set of edges, $L$ is the set of labels, $K$ is the set of keys, and $W$ is the set of values. $l_n:V\rightarrow \mathcal{P}(L)$ and $l_e:E\rightarrow \mathcal{P}(L)$ are labeling functions for nodes and edges respectively. A property is modeled as a key-value pair $(k, v), k\in K, v\in W$, and $p_n: V\rightarrow \mathcal{P}(K\times W)$ and $p_e: E\rightarrow\mathcal{P}(K\times W)$ are mappings from nodes and edges respectively to one or more properties. Note that $\mathcal{P}(S)$ is the power set of a set $S$.
The LPG model offers a flexible and intuitive way to model real-world graph structures and is used in many popular graph databases such as Neo4j~\cite{neo4j_book} or JanusGraph~\cite{janusgraph}. It also provides powerful query capabilities through languages like Cypher~\cite{francis2018cypher}.
We show an example LPG in the left part of Figure~\ref{fig:example}.

\subsection{Higher-Order Graph Data Models}

\enlargeSQ

While traditional graph data models focus on pairwise relationships, where edges only connect two nodes, many real-world systems involve HO relationships, where interactions occur among multiple entities simultaneously in more complicated ways. Higher-order graph data models thus extend plain graphs by incorporating richer structures; see Figure~\ref{fig:scheme} (top) for examples.

First, \textbf{hypergraphs} are generalizations of graphs where hyperedges connect multiple nodes instead of just two, allowing representation of group interactions (e.g., co-authorship networks, friend groups). A hypergraph $H$ is defined as $H = (V, E)$ where $V$ is the set of nodes and $E\subseteq \mathcal{P}(V)$ is the set of hyperedges.
Second, \textbf{simplicial complexes} are a specific case of hypergraphs containing structured higher-order relationships (called \emph{simplices}, \emph{simplex} in singular), such as cliques or hyperedges, such that every subset of any simplex must also be included in the complex. They are used when modeling, e.g., social interactions.
Moreover, with \textbf{subgraph collections}, graphs are complemented with additional subgraph structures that can store more fine-grained information. These subgraphs can further model nodes and be connected by higher-order subgraph-edges (e.g., large molecular graphs such as proteins). A graph with a subgraph collection is formally defined as $G = (V, E, S, F)$ where $V$ is the set of nodes, $E\subseteq V\times V$ is the set of edges, $S\subseteq \{(V', E')\ |\ V'\subseteq V, E'\subseteq E \cap (V'\times V') \}$ is the set of subgraphs, and $F\subseteq S\times S$ is the set of subgraph-edges.
Then, \textbf{tuples} can enable capturing ordered multi-node interactions in a structured way. A graph with a node-tuple collection is formally defined as $G = (V, E, T)$, where $V$ is the set of nodes, $E\subseteq V\times V$ is the set of edges, and $T\subseteq \bigcup_{k=2}V^k$ is the set of node-tuples.
Finally, \textbf{motif} is used to refer to a recurring small graph pattern (such as a triangle or a specific subgraph)~\cite{besta2024demystifying}.

\section{HIGHER ORDER: CASE STUDIES}

To further motivate integration of HO into GDBs, we now present several real-world examples of HO in graph datasets. While we cannot reveal full original data due to confidentiality issues, the following examples, as well as their graph models, are directly based on industry use cases encountered in operations. Readers uninterested in these use case details can proceed directly to Section~\ref{sec:marrying}.

\subsection{Quality Incidents as Subgraph Collections}
\label{sec:use-case-incidents}

We consider an industrial graph that models the complex operational dynamics of manufacturing plants; details are in Figure~\ref{fig:use-case}. The graph encodes interrelated entities such as production plants, batch execution schedules, plant configuration changes, and maintenance tasks. Each plant is represented as a subgraph containing its configuration states (denoted \(P\)) and the changes between those states (\(C\)). Maintenance tasks (\(M\)) are performed at specific configuration states, while production batches (\(B\)) are executed under those states.

In one observed scenario, batches \(B_1\) through \(B_7\) are executed successfully, but a quality incident occurs during batch \(B_8\). Causal analysis identifies a non-trivial correlation involving plant state \(P_3\), the preceding configuration change \(C_2\), maintenance tasks \(M_2\) and \(M_4\), and prior batches \(B_5\) and \(B_7\). This set of interacting elements forms a semantically meaningful subgraph that reflects the potential cause of the incident.

We formalize this event by lifting the correlated subgraph as a first-class HO entity. The subgraph is semantically categorized using an incident ontology through an \texttt{is\_instance\_of} relationship. Additionally, case-specific metadata -- such as severity, time of occurrence, or resolution measures -- is attached as attributes to this subgraph. Furthermore, the subgraph is connected to other structurally or semantically similar incidents occurring in different plants via similarity edges. This structure supports efficient retrieval, clustering, and analytics across multiple incident instances.

Traditional graph systems would struggle to model such a semantically rich, multi-entity incident as a coherent unit. HO-GDBs address this by treating subgraphs like the one described above as native, addressable entities. These higher-order representations enable expressive queries (e.g., \textit{``retrieve all subgraphs structurally similar to known quality incidents''}), metadata annotation, and integration with ontologies -- all while facilitating transactional and analytical operations. The seamless embedding of such structures into the HO-GDB model facilitates advanced use cases in manufacturing intelligence, incident management, and predictive maintenance.

\subsection{Causality of Plant Events as Node-Tuples}
\label{sec:use-case-causality}
    
Another example case involves a temporal state graph representing the evolving configuration of a production plant. The graph consists of devices, sensors, and control components as nodes, connected by edges indicating physical or logical dependencies. Each event occurring in the system transitions the graph from one state to another, forming a temporal sequence of modifications. The system evolution is thus captured as a sequence of graphs \( G_0, G_1, G_2, G_3, G_4 \), each reflecting a new state of the system after an event.

Initially, the graph \( G_0 = (V_0, E_0) \) contains three components \( V_0 = \{n_1, n_2, n_3\} \) and edges \( E_0 = \{e_1, e_2\} \). At event \( I_1 \), a new sensor node \( n_4 \) is installed and connected via edge \( e_3 \) to valve \( n_2 \), resulting in a new graph \( G_1 \). Following a time interval \( \Delta t_1 \), event \( I_2 \) replaces the pump node \( n_3 \) with a new pump \( n_5 \), updating the existing edge \( e_2 \) to point to the replacement. After another interval \( \Delta t_2 \), at event \( I_3 \), the operating speed \( p_1 \) of pump \( n_5 \) is adjusted to a new value \( v \). Finally, after time \( \Delta t_3 \), a quality issue is detected as event \( I_4 \).

Subsequent analysis reveals that the combination of adding the sensor (event \( I_1 \)) and adjusting the pump speed (event \( I_3 \)) -- with the specific timing \( \Delta t = \Delta t_1 + \Delta t_2 \) -- forms a causal pattern that likely contributed to the quality incident. This insight is captured in the form of a higher-order node-tuple construct \( C_1 = (I_1, I_3)\), modeling the causal relationship, with a property \(\Delta t\). Such tuples enable the expression of multi-event dependencies beyond what is representable using simple edges. They can also be associated with metadata (e.g., confidence, impact score) or linked to known incident categories via ontological typing.

Traditional graph databases lack native constructs to encapsulate multi-node, multi-event relationships with temporal semantics. The proposed HO-GDB paradigm addresses this gap by enabling node-tuples as first-class entities, allowing for structured representation of causal chains and temporal logic as higher-order patterns. These constructs can be stored, queried, and reasoned over using standard graph interfaces.

\section{MARRYING GRAPH DATABASES WITH HIGHER ORDER}
\label{sec:marrying}

We first illustrate how GDBs can be combined with HO to enable next-generation complex graph data analytics.

\begin{table}[H]
  \small
  \centering
  \begin{tabularx}{\linewidth}{lX}
    \toprule
    \textbf{Abbrev.} & \textbf{Meaning}\\
    \midrule
    HO      & Higher-Order\\
    GDB     & Graph Database \\
    GNN     & Graph Neural Network\\
    HO-GDB  & Higher-Order Graph Database \\
    HO GNN  & Higher-Order Graph Neural Network\\
    LPG     & Labeled Property Graph \\
    RDF     & Resource Description Framework \\
    ACID    & Atomicity, Consistency, Isolation, Durability\\
    OLTP    & Online Transactional Processing\\
    OLAP    & Online Analytical Processing\\
    CRUD    & Create, Read, Update, and Delete\\
    \bottomrule
  \end{tabularx}
  \caption{Glossary of abbreviations used in the paper.}
  \label{tab:abbr}
\end{table}

\subsection{How to Use Higher-Order with GDBs?}

It is not immediately clear how to natively support HO graph structures within existing GDB systems. HO constructs such as motifs, hyperedges, or subgraph collections do not naturally fit into standard vertex- and edge-based data models, and naively flattening them into first-order relationships often leads to loss of semantics, structural ambiguity, or inefficiency in querying~\cite{besta2024demystifying, morris2020weisfeiler, goh2022simplicial, giusti2023cell, feng2019hypergraph}. For instance, subgraphs or polyadic interactions cannot be easily indexed or updated using traditional neighbor-based abstractions. We therefore investigate how to represent, index, and query HO constructs while preserving compatibility with established GDB workloads and design guarantees such as ACID.

To this end, we extensively studied both the internal architecture of graph databases and the variety of HO structures used in analytics and learning, for a total of more than 300 papers. Our goal is to determine a general, expressive, and efficient method for embedding HO capabilities into GDBs, in a way that supports both OLTP and OLAP workloads, while remaining agnostic to specific HO models. Our study reveals that an approach that satisfies the above goals is to encode HO constructs as \textit{typed nodes and edges within the LPG model}, using schema extensions that transform HO entities (e.g., hyperedges) into first-class heterogeneous graph elements. This design allows each HO construct to be attributed with properties, linked to other entities, and queried through standard GDB interfaces. Moreover, this approach maintains compatibility with existing indexing, storage, and transactional semantics -- enabling seamless integration of HO graph data into production-ready GDBs.

\subsection{Graph Data Model + Higher-Order}
\label{sec:extending-gdb}

The first step towards HO-GDBs is understanding how to appropriately encode any HO structures using the underlying LPG data model used in graph databases.
%
%
Here, our key insight enabling this integration is that \textit{any HO graph can be transformed into a multi-partite heterogeneous graph}, where each partition represents a distinct HO entity (e.g., a specific hyperedge). 
%
A \textbf{heterogeneous graph} $G = (V, E, \tau, \kappa)$ comprises a set of nodes $V$, directed edges $E \subset V \times V$, a node-type function $\tau : V \rightarrow K_V$ and an edge-type function $\kappa: E \rightarrow K_E$. Heterogeneous graphs can be represented trivially in the LPG model since $\tau$ and $\kappa$ can be modeled as labels. We will show that any considered HO structure can be modeled using a heterogeneous graph \textit{without losing any information}.

Figure~\ref{fig:scheme} illustrates example HO structures (the top part) being encoded as multi-partite heterogeneous graphs (the middle part). For example, a hyperedge becomes a node labeled as a ``hyperedge'', a node-tuple becomes a node with typed edges to its elements that also store the tuple ordering in the edge properties, a subgraph is modeled as a node labeled as a ``subgraph'' connected to the nodes and edges it contains, and so on. This yields bipartite, tripartite, or even multipartite graphs, depending on the complexity of the original HO structure. These graphs are then integrated into an LPG graph (the bottom part). This enables querying these representations using standard LPG procedures. A concrete example of an LPG being augmented with HO is in Figure~\ref{fig:example}.

When transforming any HO structure into a heterogeneous graph (referred to as \textit{lowering}) and back (referred to as \textit{lifting}), we require that the transformation is \textit{lossless}, i.e., that it preserves the semantics of the original structure. This losslessness is achieved by ensuring that the transformation is \textit{isomorphism-preserving}, i.e., that any isomorphism present in the original graph is present in the result graph, and vice versa~\cite{besta2024demystifying}: two higher-order graphs $H_1$ and $H_2$ are called isomorphic, denoted $H_1\cong H_2$ if there exists a one-to-one mapping of entities from $H_1$ to $H_2$.

\begin{theorem}
    Let $L: \mathcal{G}\to\mathcal{H}$ and $L^\top:\mathcal{H}\to\mathcal{G}$ be isomorphism preserving lifting and lowering functions such that $L\circ L^\top = \text{id}_{\mathcal{H}}$ and $ L^\top\circ L = \text{id}_{\mathcal{G}}$, where $\mathcal{H}$ and $\mathcal{G}$ are respectively the space of higher-order and heterogeneous graphs. Then $L$ and $L^\top$ provide lossless transformations between higher-order and heterogeneous graphs.
    \label{thm:loss}
\end{theorem}

\iftr
\begin{proof}
    Because $L^\top$ and $L$ are mutual inverses, both maps are bijections. First, we argue about reflexivity: for every $H\in\mathcal{H}$ we have $(L\circ L^\top)(H)=H$, so lowering and then lifting returns the original higher-order graph. Now let $H_1,H_2\in\mathcal{H}$. Since $L^\top$ is isomorphism-preserving,
    \[
        H_1\cong H_2
        \;\Longleftrightarrow\;
        L^\top(H_1)\cong L^\top(H_2).
    \]
    Applying $L$ to both sides and using $L\circ L^\top=\mathrm{id}_{\mathcal{H}}$ yields
    \[
        H_1\cong H_2
        \;\Longleftrightarrow\;
        (L\circ L^\top)(H_1)\cong (L\circ L^\top)(H_2)
        \;\Longleftrightarrow\;
        H_1\cong H_2,
    \]
    showing that the transformation is lossless since isomorphism is preserved. The symmetric argument starting with $G_1,G_2\in\mathcal{G}$ shows that $L$ is lossless as well.
\end{proof}
\fi

%
\iftr
Our lifting and lowering transformations mathematically formalized for different HO graph models in Appendix~\ref{sec:appendix}.
\fi
%

\subsection{Transactions + Higher-Order}
\label{sec:oltp_concepts}

OLTP workloads in graph databases are required to satisfy the classical ACID properties: Atomicity, Consistency, Isolation, and Durability. These guarantees ensure correctness and reliability in transactional processing, particularly when dealing with complex graph operations such as insertions, deletions, and updates of nodes, edges, and now also HO constructs such as hyperedges and subgraphs. The introduction of HO into graph databases introduces new challenges for enforcing these guarantees due to the layered or decoupled representations between the logical HO entity and its physical realization in the heterogeneous graph form.
To this end, we analyze each ACID property below and illustrate how its enforcement is impacted by the presence of HO structures. In the following Section~\ref{sec:design}, we will discuss how we address these challenges with our HO-GDB design.

\subsubsection{Atomicity.\ }

Atomicity ensures that a transaction is treated as an indivisible unit of work: either all operations within the transaction are executed successfully, or none are, leaving the database unchanged in the event of failure. In HO-GDBs, a single user-level operation -- such as adding a hyperedge or modifying a subgraph -- often results in multiple low-level operations on the underlying heterogeneous representation. For instance, inserting a hyperedge may involve the creation of a new node representing the hyperedge, multiple participation edges linking it to member nodes, and potentially updates to type annotations. To preserve atomicity, all of these operations must be bundled into a single atomic transaction. If any of these steps fail, the database engine must revert all intermediate changes, ensuring that no partial HO entity remains in the system.
This is particularly challenging in distributed settings or when supporting asynchronous execution. Systems such as Neo4j support transaction scopes that can be extended to wrap these HO operations through custom transactional APIs or triggers.
%

\subsubsection{Consistency.\ }

Consistency ensures that each transaction transitions the database from one valid state to another, upholding all schema constraints and invariants.
In traditional GDBs, one distinguishes \textbf{referential consistency}~\cite{atlanConsistency2022, mongodbConsistency} (when an edge or a vertex is added or deleted, all references to that edge -- from indexes, properties, or related nodes -- must be correctly updated or removed to prevent dangling references) and \textbf{structural consistency} (modifying a node must maintain appropriate updates to its incident edges or properties). 

In HO-GDBs, other forms of consistency are also needed. We first identify \textbf{HO invariance consistency}, which ensures that semantical invariants of HO structures are maintained. For example, a hyperedge must connect at least two nodes, or a motif must satisfy a given prescribed structural pattern. Updates that violate these rules must be rejected or corrected before commit. Second, \textbf{HO cross-representation consistency} ensures that a lifted HO construct (e.g., a hyperedge node with participation edges) must remain logically equivalent to its abstract HO definition; \textit{this is emergent from the losslessness of the lifting and lowering operations}. For example, renaming a subgraph node label must reflect across all its nested references.
These forms of consistency may require extended integrity constraints or declarative schema specifications akin to constraint languages over graph schemas~\cite{angles2017foundations}. 

\subsubsection{Isolation.\ }

Isolation ensures that concurrent transactions do not interfere with one another and that their effects are equivalent to some serial execution. Isolation is hard to achieve with HO due to potential interleaved updates on multiple related elements. For example, two concurrent transactions may attempt to add overlapping motifs or update shared subgraphs. If isolation is not enforced, partial updates from one transaction might corrupt the semantic validity of the other's HO construct.
Isolation must therefore be enforced not only at the level of individual nodes and edges, but across semantically grouped elements representing a single HO construct. This may require lock management or conflict resolution protocols extended to HO entities. 

\subsubsection{Durability.\ }

Durability ensures that once a transaction commits, its effects persist even in the face of failures.
In HO-GDBs, durability extends to complex transactional groups that realize a single HO structure. This includes the persistence of all elements involved in the lifted representation: HO nodes, edges linking them to participants, metadata, and possible index updates.
A key challenge arises in ensuring that partially persisted HO constructs (e.g., a hyperedge with missing participants due to a crash) do not remain in the system.
Furthermore, because HO constructs can be nested or hierarchical (e.g., a subgraph containing motifs that include hyperedges), durability must be recursively enforced across these levels. Failure to do so may lead to structural fragmentation or data corruption.

\section{SYSTEM DESIGN \& IMPLEMENTATION}
\label{sec:design}

We now describe the architecture and implementation of our prototype of an HO-GDB system for HO graph modeling and processing within a traditional graph database. It is designed as a lightweight, modular, parallelizable extension to existing graph databases, with Neo4j as our chosen backend. The system supports lossless lowering and lifting between various higher-order structures such as hypergraphs, node-tuples as well as subgraphs, and their heterogeneous representations. It ensures transactional integrity, supports expressive querying, and provides a scalable API for selected OLTP and OLAP workloads involving higher-order graph models.


\subsection{Design Overview}
\label{sec:design-overview}

The system design is organized into three decoupled layers, enabling modular and scalable development: the Higher-Order API layer, the Higher-Order Processing layer, and the native GDB backend.


First, the \textbf{Higher-Order API layer} implements an interface to perform OLTP and OLAP procedures on HO graphs. 
This layer is the entry point for applications and scripts interacting with HO-GDB. It encodes the logic for manipulating higher-order entities and their transformations into LPG compatible structures. This layer also provides a semantically meaningful API for a high-level set of core methods that can perform CRUD (create, read, update, delete) operations on different higher-order structures, by combining several subtransactional methods within a single transaction, ensuring both atomicity across complicated reads and writes as well as correctness. The query API also comes with a set of lightweight classes, from Label and Property to graph element classes such as Node, Edge, HyperEdge, Subgraph, and NodeTuple to model LPG-compliant and other HO entities, which encapsulate structural metadata and improving expressivity.


The \textbf{Higher-Order Processing layer} manages the execution of LPG-based queries on the graph database, bridging the HO API and the underlying graph database, by scheduling low-level, subtransactional graph operations within HO queries over a backend graph database. It exposes a clean LPG-compatible API for adding, deleting, updating, and matching nodes and edges, along with importing/exporting structures from files. This layer also provides functionality to create and close database sessions and transactions, run transactional and non-transactional queries as well as index creation and deletion. 


The \textbf{backend graph database} is the underlying database engine on which the transformed data structures will be stored and queried. Thus, this layer handles the actual LPG graph storage and query execution.

\subsection{ACID Guarantees}
\label{sec:acid}

We now discuss how the implementation and design maintains the ACID guarantees.

\subsubsection{Atomicity.}

All higher-order operations (such as inserting, deleting, or updating) are internally decomposed into multiple low-level subtransactional database operations, which create typed auxiliary nodes or edges within a transaction. Atomicity is enforced by bundling these low-level queries into a single transaction, making sure that either all of them succeed or none. If any sub-operation fails due to a system failure or constraint violation, partial updates are rolled back. We harness fine-grained concurrency control implemented within the underlying backend—e.g., via multi-version concurrency control (MVCC) or subgraph-level locking—to help reduce contention while preserving correctness~\cite{besta2023demystifying}.

\subsubsection{Consistency.}

HO-GDB enforces several levels of consistency. First, each higher-order data object is validated at creation, enforcing \textbf{invariant consistency}. For this, HO-GDB's auxiliary metadata and HO graph model classes come with assertions executed at query initialization.
Second, \textbf{cross-representation consistency} is maintained by the atomicity guarantees of the API layer and the isomorphism preservation provided by the lifting and lowering algorithms (Section~\ref{sec:extending-gdb}).
Moreover, both \textbf{referential} and \textbf{structural consistency} are maintained through HO-GDB's delete policy. When a higher-order entity (such as a subgraph or tuple) is deleted, all incident relationships (e.g., node incidence, edge incidence) are automatically removed, to prevent dangling links and ensure consistency. Here, we ensure that deleting constituent elements (such as nodes within a tuple or subgraph) does not result in the automatic deletion of the HO entity itself (i.e., as this HO entity is separate from its constitutent elements). For example, removing all nodes from a node-tuple will leave behind an empty tuple node with its labels and properties preserved. In the underlying heterogeneous structures, we leverage Neo4j's referential and structural consistency, which ensures that there are no dangling references or hanging edges in the lowered representation stored on Neo4j; the losslessness of the transformations guarantees that these representations will always lift into valid higher-order graphs.


\subsubsection{Isolation \& Durability.}

Both isolation and durability benefit from the guarantees of the underlying backend, in this case Neo4j. Isolation is ensured through a fine-grained locking mechanism to ensure data consistency during concurrent write operations.
Durability guarantees ensure that once an HO transaction is committed, the changes are written to disk and are resilient to crashes or system failures.

\subsection{Indexing}

HO-GDB's Higher-Order API layer also provides an interface for managing indexes over HO structures based on (label, property) pairs. This enables users to define and manage indexes over any HO entities analogously to how this is done for standard nodes and edges, by type and property regardless of how these are implemented in the underlying data model.
Since HO-structures are lowered to LPG graphs, this implementation ensures that indexing is consistent across any graph structures.

\subsection{Seamless Cloud Deployment}

To support remote or containerized deployments and to further enhance modularity, HO-GDB also implements an optional proxy driver that abstracts database connectivity through a REST-based proxy service modeled after the Neo4j Python client driver. This layer decouples the client process from the graph database driver by routing Cypher queries, session management, and transaction management through a lightweight HTTP interface. It enables seamless deployment in cloud, serverless, or multi-instance environments without the user having to modify core database logic since it ensures compatibility with the rest of the system through a plug-and-play session interface identical to the native driver.


\subsection{Seamless Integration with Graph Stores \& Applications}

When integrating HO-GDB with other graph stores and applications, there is no need for any applications changes, because of our decision of bridging the higher-order with GDBs at the level of the LPG data model, which enables seamless and portable integration. Further, no changes to transactional APIs, indexing, or query optimizers are needed. Our backend-agnostic HO-API Layer can be ported to any LPG store by implementing a thin interface that maps basic CRUD operations to the store’s query language (cf.~Section~\ref{sec:design-overview}).

\section{HIGHER ORDER IN GDB WORKLOADS}

We now discuss developing HO-GDB workloads.

\subsection{OLTP Workloads}
\label{subsec:OLTP}

HO-GDB provides an OLTP interface for processing standard and HO graph structures. This interface is meant to extend existing LPG query interfaces to higher-order structures, while maintaining semantic correctness. It supports all four CRUD operations over all considered HO entities, and can be straightforwardly extended to any other HO concepts. Each such entity is implemented as an abstraction with a unique schema comprising sub-entities, labels, and properties (an example is shown in Listing~\ref{lst:sg}). The HO query interface abstracts away the underlying lifting and lowering and enables declarative OLTP procedures on these entities. 

\begin{lstfloat}[t]
\begin{lstlisting}[language=Python, caption=\textmd{Definition of the Subgraph HO data model and selected routines in the interface provided by an HO-GDB. The proposed API integrates seamlessly into nearly all existing graph databases because it expresses HO structures as simplified LPG graphs.}, label=lst:sg]
class Subgraph:
  subgraph_nodes: List[Node]
  subgraph_edges: List[Edge]
  labels: List[Label]
  properties: List[Property]

gdb.add_subgraph(
  Subgraph(
    subgraph_nodes = [n1, n2, n3],
    subgraph_edges = [e12, e23, e31],
    labels = [Label("Triangle")],
    properties = [Property("name", str, "triangle1")]
  )
)

gdb.delete_subgraph(
  Subgraph(
    subgraph_nodes = [], subgraph_edges = [],
      [Label("Triangle")], [Property("name", str, "triangle1")]
  ) # Parameters optional, but adding info improves query efficiency and correctness
) 
\end{lstlisting}
\end{lstfloat}

As in Cypher~\cite{francis2018cypher}, deletion and update operations do not require full structural specification. To support this, HO-GDB performs internal pattern matching based on the provided labels and properties. This internal matching mechanism identifies the minimal subgraph(s) in the database that satisfy the given label and property constraints, without requiring the user to specify the complete structure to be deleted. This allows for concise queries that can target specific nodes or relationships while still maintaining semantic correctness and safety of the operations.


\subsection{OLAP Workloads}

The versatile HO CRUD API enables easy development of OLAP workloads.
For example, HO-GDB enables semantically rich path traversal queries over HO structures by abstracting paths as declarative Path objects composed of labeled graph elements and property constraints. These paths are internally transformed, using the same lowering algorithms, into graph pattern matching queries that operate over the lowered heterogeneous representation, allowing users to implement HO traversal algorithms such as Hypergraph walks, an important class of HO analytics~\cite{hayashi2020hypergraph, lu2013high, chun2024random, eriksson2021choosing, luo2024sampling, mulas2022random, lee2011hyper, niu2019rwhmda, chitra2019random, carletti2020random}, see Listing~\ref{lst:path}.
Another example important class of OLAP workloads enabled by the proposed HO-GDB paradigm are higher-order GNNs (HO GNNs).

GNNs have emerged as powerful tools for learning representations of graph data~\cite{wu2020comprehensive, zhou2020graph, zhang2020deep,
chami2020machine, hamilton2017representation, bronstein2017geometric, kipf2016semi, xu2018powerful, wu2019simplifying, besta2023parallel, bazinska2023cached, besta2023high}. GNNs perform tasks like node classification, link prediction, and graph classification by effectively capturing relational dependencies in graphs using message-passing paradigms to propagate information throughout the graph, unlike traditional machine learning models, which struggle to take advantage of the graph structure.

While standard GNNs are designed for pairwise relational graphs, recent advances have extended them to handle HO. Such HO GNNs~\cite{besta2024demystifying} often outperform traditional GNNs on tasks requiring nuanced structural understanding. Moreover, different formulations of higher-order GNNs, such as hypergraph networks, simplicial complexes, subgraph-based models, or more complicated structures, allow for tailored architectures that fit specific datasets better than traditional GNNs~\cite{morris2020weisfeiler, bunch2020simplicial, yang2022simplicial, goh2022simplicial, giusti2023cell, heydari2022message, feng2019hypergraph, bai2021hypergraph, hajij2020cell, wang2020gognn, bouritsas2022improving, rossi2018hone, qian2022ordered, xu2018powerful, morris2019weisfeiler, hajij2022topological, bodnar2021cellular, bodnar2021topological, arya2020hypersage, yadati2020neural}.
In HO-GDB, we implement an HO-GNN based on~\cite{fey2020hierarchicalintermessagepassinglearning} by jointly operating on molecular graphs and their junction tree decompositions~\cite{jin2019junctiontreevariationalautoencoder}. This HO model improves expressiveness by capturing structural patterns and enabling message passing between both standard vertices and HO subgraphs grouping these vertices.


\begin{lstfloat}[t]
\begin{lstlisting}[language=Python, caption={\textmd{Find sorted list of names of ABC university research groups that have collaborated with Alice. In the graph model, belonging to a research group is represented as a subgraph node-membership.}}, label=lst:path]
path = Path()
path.add(
  n1 = Node([Label("Person")],
    [Property("name", str, "Alice")])
)
path.add(
  s1 = Subgraph([Label("ResearchGroup")],
      [Property("univ", str, "ABC")])
)
gdb.traverse_path(
  [path],return_values=["s1.name"],sort=["s1.name"]
)
\end{lstlisting}
\end{lstfloat}

\begin{lstfloat}[t]
\begin{lstlisting}[language=Python, caption={\textmd{Constructing a torch_geometric Data object by retrieving node features, edge indices, and edge attributes using path traversal queries on an HO-GDB.}}, label=lst:torch]
# Get Node features
path_x = Path()
path_x.add(x = Node([Label("Node")]))
x = torch.tensor(gdb.traverse_path(
  [path_x],return_values=["x.feat"],sort=["x.id"]
)) # works for any HO model

# Get Edge Index and Attributes
path_e = Path()
path_e.add(s = Node([Label("Node")]))
path_e.add(e = Edge(label=Label("Edge")))
path_e.add(t = Node([Label("Node")]))
edge_index = torch.tensor(gdb.traverse_path(
  [path_e],return_values=["s.id","t.id"],
    sort=["s.id","t.id"]
))
edge_attr = torch.tensor(gdb.traverse_path(
  [path_e],return_values=["e.feat"],
    sort=["s.id","t.id"]
))
data = torch_geometric.Data(x, edge_index, edge_attr)
\end{lstlisting}
\end{lstfloat}



\section{THEORETICAL ANALYSIS}

\begin{table*}[ht]
    \centering
    \setlength{\tabcolsep}{3pt}
    \small
    \begin{tabular}{lccccclcc}
        \toprule
        & \multicolumn{1}{c}{\textbf{Lowering \& Lifting}} & \multicolumn{2}{c}{\textbf{Node}} & \multicolumn{2}{c}{\textbf{Edge}} & \multicolumn{3}{c}{\textbf{HO Structure}} \\
        \cmidrule(lr){3-4} \cmidrule(lr){5-6} \cmidrule(lr){7-9}
        \textbf{GDB Model} & \textbf{Time \& Storage} & \textbf{Insert} & \textbf{Delete} & \textbf{Insert} & \textbf{Delete} & \textbf{Structure} & \textbf{Insert} & \textbf{Delete} \\
        \midrule
        Plain Graph & -  & $O(d)$ & $O(d+\delta)$ & $O(d)$ & $O(d)$ & - & - & - \\
        Hypergraph & $O(d(nn_H))$ & $O(d)$ & $O(d+\delta)$ & - & - & Hyperedge & $O(d+n)$ & $O(d+n)$ \\
        Simplicial Complexes & $O(d(nn_H))$ & $O(d)$ & $O(d+\delta)$ & - & - & Hyperedge & $O(d+n)$ & $O(d+n)$ \\
        Node-Tuple Collections & $O(d(m + nn_T))$ & $O(d)$ & $O(d+\delta)$ & $O(d)$ & $O(d)$  & Tuple & $O(d+n)$ & $O(d+n)$ \\
        Subgraph Collections & $O(d(nn_S+ n_F+mn_S))$ & $O(d)$ & $O(d+\delta)$ & $O(d)$ & $O(d)$ & \makecell[l]{Subgraph \\ Subgraph-Edge} & \makecell{$O(d+n+m)$ \\ $O(d)$} & \makecell{$O(d+n+m)$ \\ $O(d)$} \\
        \bottomrule
    \end{tabular}
    \caption{\textmd{Time complexity of lifting and lowering procedures and OLTP operations for different HO graphs with $n$ nodes, $m$ edges, $n_H$ hyperedges, $n_T$ node-tuples, $n_S$ subgraphs, and $n_F$ subgraph-edges. Here, every entity and relationship has $O(d)$ properties, nodes and subgraphs have maximum degree of $\delta$, and hyperedges/node-tuples/subgraphs contain $O(n)$ nodes.}}
    \label{tab:oltp_complexity}
\end{table*}

We present a theoretical analysis of operations in HO-GDB when harnessing different HO data models. We consider insertions and deletions for basic graph primitives (nodes and edges). The results are in Table~\ref{tab:oltp_complexity}.
\iftr
Full derivations are in Appendix~\ref{sec:appendix}.
\fi


%
These results for the \textbf{time complexity of OLTP operations} reflect the cost of maintaining heterogeneously encoded HO constructs, where each HO operation is decomposed into primitive LPG operations such as node/edge additions and index maintenance. Deletion costs are typically higher due to the need to remove incident references.
%
%
We also analyze the \textbf{time complexity of the lowering and lifting} operations as well the \textbf{storage complexity} between HO graphs and their LPG representations.
These transformations are shown to be \emph{lossless and isomorphism-preserving}, as proven in Section~\ref{sec:extending-gdb}. 
%
%
Our theoretical analysis confirms that HO-GDB's higher-order transformations support bounded time complexity for OLTP operations. The lifting and lowering procedures introduce overhead that scales linearly with the size and count of the HO constructs, which remains practical for real-world deployments. These complexity bounds underpin HO-GDB’s ability to offer competitive performance while enabling expressive higher-order graph workloads.

\section{EVALUATION}

We now evaluate the performance of the proposed HO-GDB. The main goal is to show that (1) harnessing HO does enhance accuracy of workloads such as graph learning, while simultaneously (2) does not incur significant overheads in the latency and throughput of graph queries.

\subsection{Experimental Setup \& Methodology}

\ifnblnd
All experiments are conducted on the Helios hybrid cluster at PLGrid, running SUSE Linux Enterprise Server 15 SP5 on an \texttt{aarch64} architecture. Each node is equipped with 4×72-core NVIDIA Grace CPUs (288 cores total) clocked up to 3.5 GHz, with 480GB RAM and large multi-level caches (\textsc{L3}: 456 MiB). We utilize an NVIDIA GH200 GPU with 120GB memory and CUDA 12.7 for training the higher-order GNN, under driver version 565.57.01. The host kernel is 5.14.21 and system management is provided by CrayOS. We also deploy a Neo4j 5.26.5 Community Edition instance. All workloads are run using Python v3.11.5, with Neo4j accessed via the official Neo4j Python driver v5.28.1.
\else
All experiments are conducted on a cluster running SUSE Linux Enterprise Server 15 SP5 on an \texttt{aarch64} architecture. Each node is equipped with 4×72-core NVIDIA Grace CPUs (288 cores total) clocked up to 3.5 GHz, with 480GB RAM and large multi-level caches (\textsc{L3}: 456 MiB). We utilize an NVIDIA GH200 GPU with 120GB memory and CUDA 12.7 for training the higher-order GNN, under driver version 565.57.01. The host kernel is 5.14.21 and system management is provided by CrayOS. We also deploy a Neo4j 5.26.5 Community Edition instance. All workloads are run using Python v3.11.5, with Neo4j accessed via the official Neo4j Python driver v5.28.1.
\fi
To minimize overhead for the Neo4j-backend, we use Bolt to communicate with Neo4j and expose a client-side Python-API. However, HO-GDB can be used with other communication protocols.

We experiment with running Neo4j on (1) the same node and on (2) different nodes than the node where the queries are made.
\iftr
We only plot the results for scenario~(1) to be concise.
\fi
The advantages of HO-GDB are even more prominent for scenario~(2) than in scenario~(1) because HO-GDB batches HO requests, minimizing inter-node latencies (we conservatively plot the results for scenario~(1)).

\subsection{Challenges in Selecting Backend GDBs}

Many graph databases exist, but most are not freely accessible. We attempted to obtain systems like Oracle's PGX, but were unsuccessful. Among open systems, we selected those supporting both OLTP and OLAP. After thorough evaluation and setup, we successfully configured Neo4j using its in-memory mode.

Overall, we focus on single-node query processing with potential replication for higher query throughput because the available actively maintained GDBs are either single-node designs only, or offer distributed-memory implementations hidden behind steep paywalls (costs easily scale to several 10,000 USD per evaluation).

\subsection{Scalability of HO OLTP Workloads}

\begin{figure*}
    \begin{subfigure}[t]{0.33\textwidth}
        \centering
        \includegraphics[width=\textwidth]{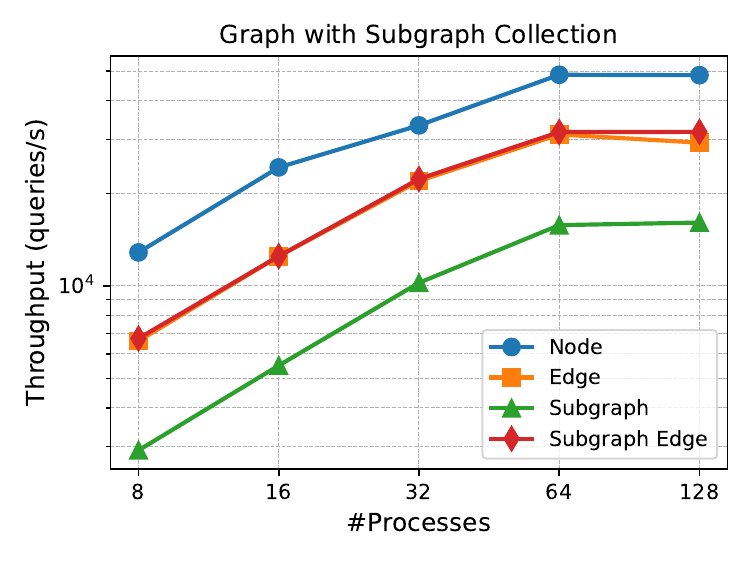}
        \caption{\textmd{Strong Scaling (subgraphs, reads)}}
        \label{fig:plot_subg_strong}
    \end{subfigure}
    ~
    \begin{subfigure}[t]{0.33\textwidth}
        \centering
        \includegraphics[width=\textwidth]{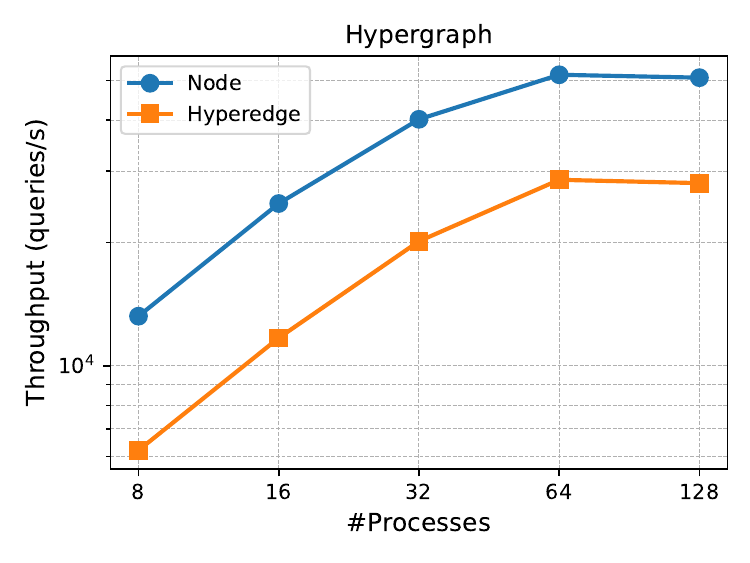}
        \caption{\textmd{Strong Scaling (hypergraph, reads)}}
        \label{fig:plot_hg_strong}
    \end{subfigure}
    ~
    \begin{subfigure}[t]{0.33\textwidth}
        \centering
        \includegraphics[width=\textwidth]{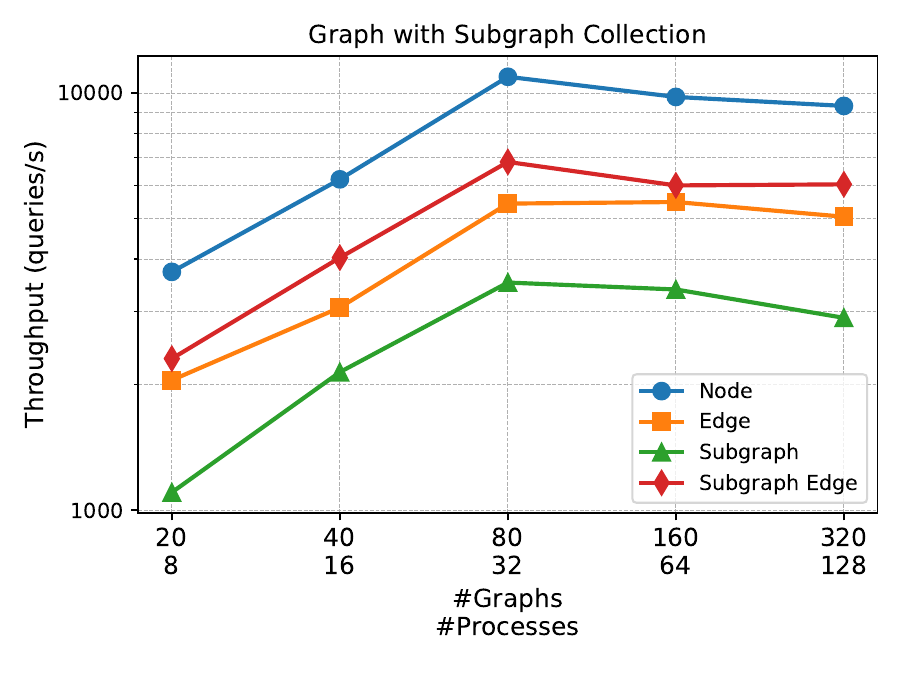}
        \caption{\textmd{Weak Scaling (subgraphs, reads)}}
        \label{fig:plot_weak}
    \end{subfigure}
    
    \begin{subfigure}[t]{0.25\textwidth}
        \centering
        \includegraphics[width=\textwidth]{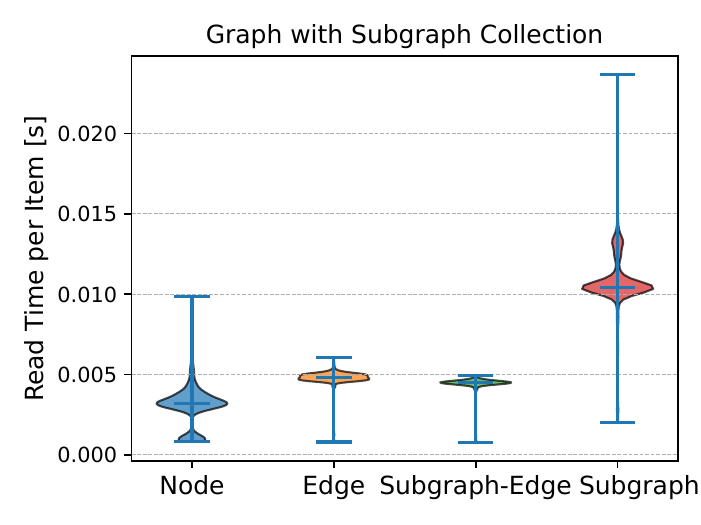}
        \caption{\textmd{Read latency (subgraphs)}}
        \label{fig:plot_sg_violin}
    \end{subfigure}
    ~
    \begin{subfigure}[t]{0.25\textwidth}
        \centering
        \includegraphics[width=\textwidth]{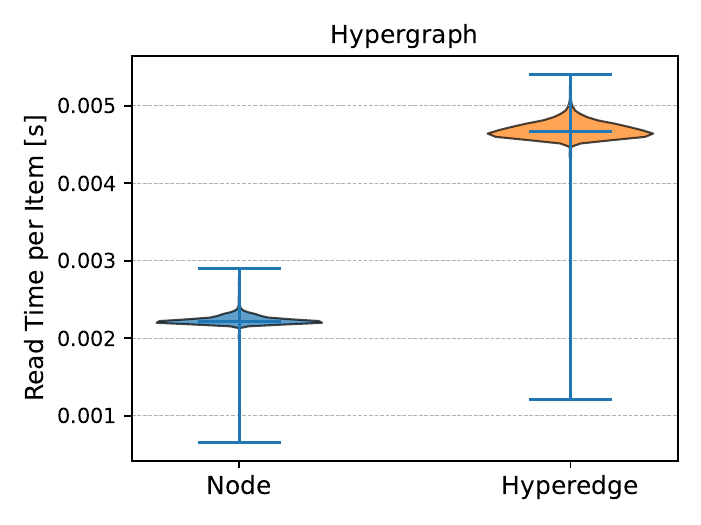}
        \caption{\textmd{Read latency (hypergraphs)}}
        \label{fig:plot_hg_violin}
    \end{subfigure}
    ~
    \begin{subfigure}[t]{0.25\textwidth}
        \centering
        \includegraphics[width=\textwidth]{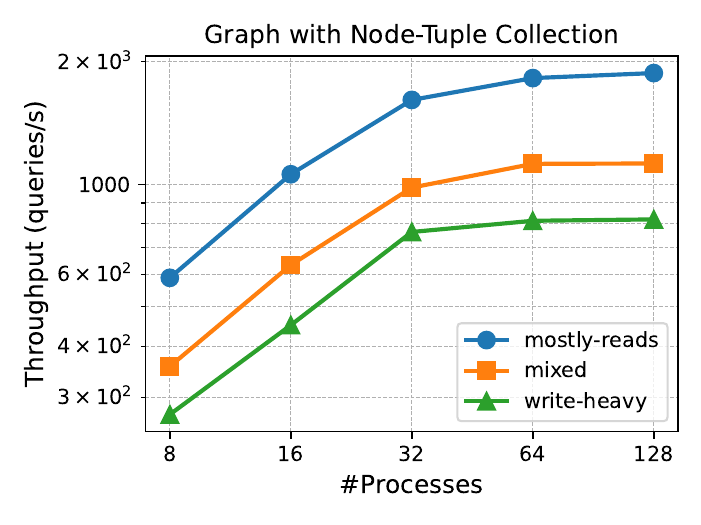}
        \caption{\textmd{Strong scaling (node-tuples, mixed)}}
        \label{fig:plot_tuples}
    \end{subfigure}
    ~
    \begin{subfigure}[t]{0.25\textwidth}
        \centering
        \includegraphics[width=\textwidth]{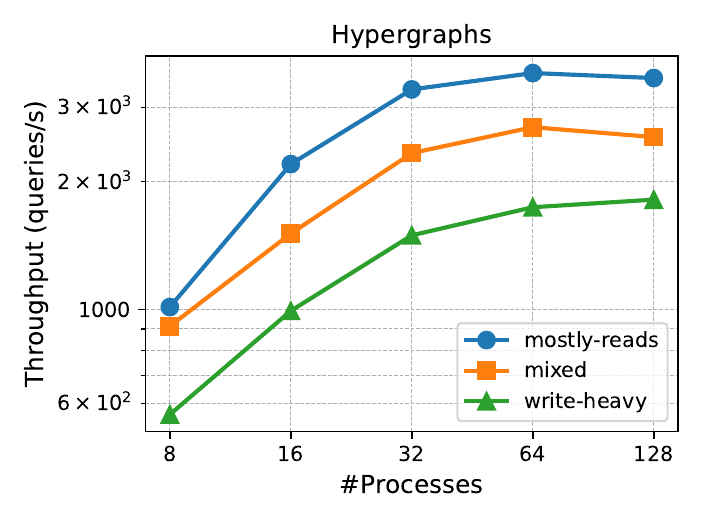}
        \caption{\textmd{Strong scaling (hypergraphs, mixed)}}
        \label{fig:plot_hg_oltp}
    \end{subfigure}
\vspace{-1em}
    \caption{\textmd{Higher-Order OLTP results.}}
\end{figure*}


We first analyze OLTP workloads by stressing our system with a high-velocity stream of HO graph queries and transactions, while scaling both parallel processes for the fixed graph size (\textbf{strong scaling}) as well as scaling the graph dataset size with the number of processes (\textbf{weak scaling}).
We evaluate the parallel performance of HO-GDB by measuring the total read time and throughput for 4,096 entities (nodes, edges, subgraphs, and subgraph-edges) using up to 128 concurrent processes (Figure~\ref{fig:plot_subg_strong}), from a database containing the first 400 molecules of the ZINC dataset~\cite{irwin2020zinc20} with a junction tree transformation~\cite{jin2019junctiontreevariationalautoencoder}. The HO graph has 9,137 nodes, 19,622 edges, 5,692 subgraphs, 10,584 subgraph-edges; the lowered graph contains 45,035 nodes and 159,936 edges. 
%
%
We observe that all entity types exhibit strong scaling, with total read time decreasing linearly as the number of processes increase. This confirms HO-GDB's ability to efficiently parallelize HO queries. Figure~\ref{fig:plot_weak} shows weak scaling of read queries in graphs with subgraph collections, where the number of processes and dataset size increase; the patterns are similar.
Next, Figure~\ref{fig:plot_hg_strong} shows strong scaling of read queries on the MAG-10 dataset~\cite{amburg2020categorical, Sinha-2015-MAG}, a subset of the Microsoft Academic Graph with labeled hyperedges. The HO graph contains 80,198 nodes and 51,889 hyperedges of 10 different types, and the lowered representation of this graph contains 132,087 nodes and 180,725 edges.
We observe a drop in scaling between 64 and 128 processes, especially pronounced for subgraph, hyperedge, and subgraph-edge reads, and attribute this to the saturation of the CPU on the Neo4j instance, indicating that HO queries, especially on complex structures like subgraphs, are more computationally intensive and sensitive to resource limits. 

For further insights, we also study per-item read time distributions across HO entity types, see Figures~\ref{fig:plot_sg_violin} and~\ref{fig:plot_hg_violin}. As expected, reads are fastest for smaller entities that contain a smaller representation in lowered heterogeneous graphs, such as nodes, and slowest for more complex entities such as subgraphs. The consistent skew and long tails, especially for nodes, suggest that early query runs may have incurred higher latency due to cold caches.




\begin{figure*}[t]
    \centering
    \includegraphics[width=\textwidth]{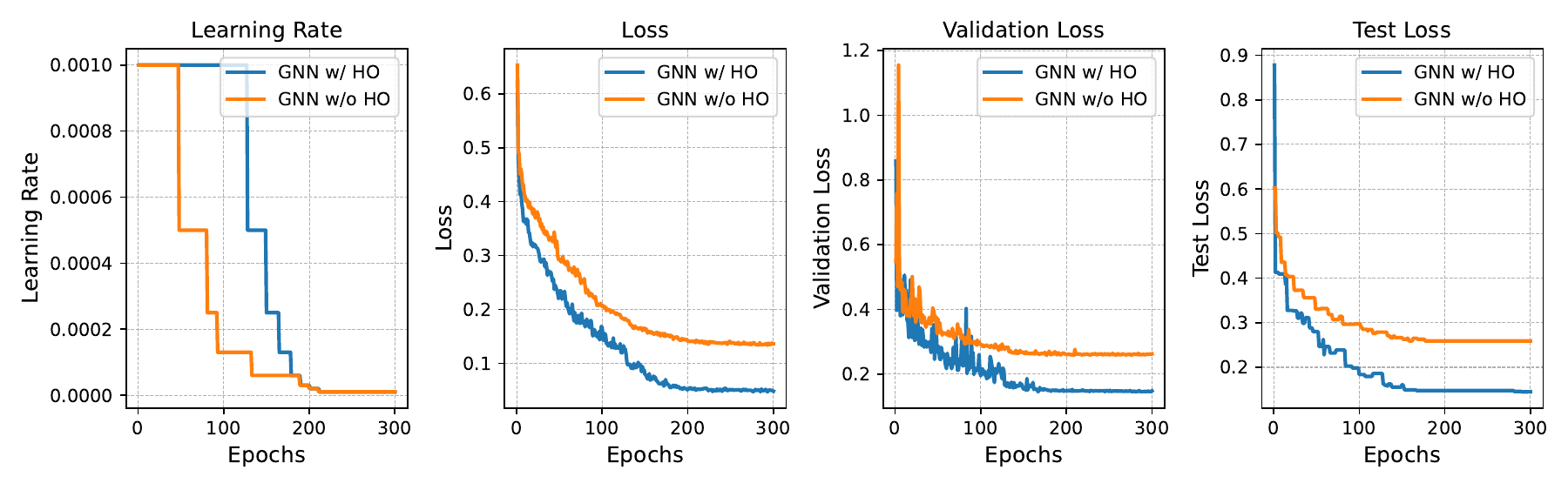}
\vspace{-2.5em}
    \caption{\textmd{Higher-Order OLAP results (HO GNNs).}}
    \label{fig:plot_gnn}
\end{figure*}

\begin{table}[t]
    \centering
    \small
    \begin{tabular}{lccc}
        \toprule
        \textbf{Query Type} & \textbf{mostly-reads} & \textbf{mixed} & \textbf{write-heavy} \\
        \midrule
        \textbf{Reads} & \textbf{95\%} & \textbf{50\%} & \textbf{25\%} \\
        Retrieve node & 31.67\% & 16.67\% & 8.33\% \\
        Retrieve edge & 31.67\% & 16.67\% & 8.33\% \\
        Retrieve node-tuple & 31.67\% & 16.66\% & 8.33\% \\
        \midrule
        \textbf{Writes} & \textbf{5\%} & \textbf{50\%} & \textbf{75\%} \\
        Add edge & 4.16\% & 41.67\% & 62.5\% \\
        Add node-tuple & 0.84\% & 8.33\% & 12.5\% \\
        \bottomrule
    \end{tabular}
    \caption{\textmd{Breakdown of OLTP query workloads used in the node-tuple strong scaling experiment.}}
    \label{tab:oltp-tuple}
        \vspaceSQ{-1.5em}
\end{table}

\begin{table}[t]
    \centering
    \small
    \begin{tabular}{lccc}
        \toprule
        \textbf{Query Type} & \textbf{mostly-reads} & \textbf{mixed} & \textbf{write-heavy} \\
        \midrule
        \textbf{Reads} & \textbf{93.75\%} & \textbf{50\%} & \textbf{25\%} \\
        Retrieve node & 46.87\% & 25\% & 12.5\% \\
        Retrieve hyperedge & 46.88\% & 25\% & 12.5\% \\
        \midrule
        \textbf{Writes} & \textbf{6.25\%} & \textbf{50\%} & \textbf{75\%} \\
        Update node & 3.13\% & 42.19\% & 57.81\% \\
        Add hyperedge & 3.12\% & 7.81\% & 17.19\% \\
        \bottomrule
    \end{tabular}
    \caption{\textmd{Breakdown of OLTP query workloads used in the hypergraph strong scaling experiment.}}
    \label{tab:oltp-hg}
        \vspaceSQ{-1.5em}
\end{table}


We also analyze performance across a range of OLTP query mixes to evaluate how graphs with node tuples and hypergraphs scale under varying read-to-write workloads and concurrent processes. Figure~\ref{fig:plot_tuples} shows strong scaling throughput for three representative query profiles of 12,288 queries (mostly-reads, mixed, and write-heavy, as defined in Table~\ref{tab:oltp-tuple}) executed over a randomly generated graph with 20k nodes, 997.5k edges and 4,000 node-tuples, with a tail-heavy edge and tuple distribution (note that the used LPG graphs have rich label \& property sets attached to nodes as well as edges, so their counts of vertices and edges are necessarily smaller than with standard graphs without such rich data).
The size of node tuples varied between 2 and 984, with an average of 69. The lowered graph comprises 1.02M nodes and 2.27M edges.
Similarly, we run three representative query profiles of 8096 queries (Table~\ref{tab:oltp-hg}) on a random hypergraph with 60k nodes and 57k hyperedges; the lowered graph comprises 117k nodes and 1.07M edges. Results are shown in Figure~\ref{fig:plot_hg_oltp}. We observe similar scaling patterns with higher throughput.

As expected, read-dominant workloads offer significantly higher throughput and benefit significantly from increased parallelism. Meanwhile, the mixed and write-heavy profiles, while scaling pretty well until 64 processes, exhibit diminishing returns beyond 64, reflecting the higher overheads introduced by write queries in the underlying database (in large-scale graph and irregular workloads, parallel performance often plateaus or degrades beyond $\approx$32-64 processes/node due to issues such as hardware saturation~\cite{green2019performance, mirsoleimani2015scaling, schweizer2015evaluating}). Nevertheless, these results demonstrate that HO-GDB can efficiently scale OLTP workloads over higher-order graph data.

Note that the cost of write operations is not inherently bound to the HO-GDB design, but largely depends on user choices for consistency/isolation and the underlying GDB implementation. One can increase the throughput of HO writes by, e.g., lowering isolation guarantees. This, however, could result in phantom reads and related effects known from relational databases. Currently we conservatively maintain full consistency/isolation with serializability.

\subsection{Enhanced Accuracy of HO OLAP Workloads}

The key potential of harnessing HO is to improve the accuracy of OLAP workloads. For this, we investigate graph-level prediction tasks with Higher-Order Graph Neural Networks (HO GNNs), implementing the Inter-Message Passing-enabled GNN architecture described in~\cite{fey2020hierarchicalintermessagepassinglearning}. This is a dual-message-passing architecture with two GNNs (molecular graph and its junction tree). These networks exchange features between layers, allowing cross-level interaction. The base GNN uses standard Graph Isomorphism Operator layer (GINEConv)~\cite{hu2020strategiespretraininggraphneural}.

Our pipeline begins by populating the graph database with 400 molecules from the ZINC dataset’s test suite; these graphs are then queried and used in training the model. We compare the performance of the resulting HO GNN against a traditional non-HO GINEConv architecture with three layers.

As shown in Figure~\ref{fig:plot_gnn}, the HO-enhanced GNN exhibits consistently faster convergence and significantly improved loss across training, validation, and test phases. The test loss is reduced by 43.99\%, clearly demonstrating that higher-order representations enable the model to capture complex structural dependencies that are invisible to standard edge-based GNNs. Specifically, HO GNN leverages the additional parameters at the junction tree level, capturing complex subgraph structures and higher-order dependencies. Non-HO fails to extract these multi-scale features, leading to inferior performance. The learning rate schedule also reveals that the HO-enhanced model tolerates higher initial learning rates for longer periods, indicating a more stable and informative gradient landscape during training.


\subsection{Preprocessing Costs}

We also consider costs of preprocessing. Here, lifting and lowering transformations are one-time preprocessing overheads; they take $<$5\% of the runtime of conducted OLTP experiments and $<$1\% of the OLAP jobs (modern GDBs are optimized for batch insertions).

\section{RELATED WORK}

Our work touches upon many different areas in system design, data analytics, and machine learning.

\textbf{Graph databases} have been researched in both academia and
industry~\cite{angles2018introduction, davoudian2018survey, han2011survey,
gajendran2012survey, gdb_survey_paper_Kaliyar, kumar2015domain,
gdb_survey_paper_Angles}, in terms of query languages~\cite{angles2018g,
bonifati2018querying, angles2017foundations}, database
management~\cite{gdb_management_huge_unstr_data, pokorny2015graph,
junghanns2017management, bonifati2018querying, miller2013graph}, execution in
novel environments such as the serverless setting~\cite{toader2019graphless,
mao2022ermer}, novel paradigms~\cite{besta2022neural}, and others~\cite{deutsch2020aggregation}.  Many
graph databases exist~\cite{tigergraph2022ldbc, cge_paper, tiger_graph_links,
janusgraph, chen2019grasper, azure_cosmosdb_links, amazon_neptune_links,
virtuoso_links, arangodb_links, tesoriero2013getting, profium_sense_links,
triplebit_links, gbase_paper, graphbase_links, graphflow, livegraph,
memgraph_links, dubey2016weaver, sparksee_paper, graphdb_links,
redisgraph_links, dgraph_links, allegro_graph_links, apache_jena_tbd_links,
mormotta_links, brightstardb_links, gstore, anzo_graph_links, datastax_links,
blaze_graph_links, oracle_spatial, stardog_links, besta2023graph,
cayley_links}.
More recently, new systems have emerged, focusing on scalability and integration with analytics and streaming. Examples include TuGraph~\cite{tugraph}, Ultipa Graph~\cite{ultipa}, and NebulaGraph~\cite{nebulagraph}, which support massive parallelism and in some cases hybrid architectures combining transactional and analytical graph workloads. TuGraph, for instance, supports both ACID-compliant transactions and graph analytics via a shared in-memory core. Ultipa and Memgraph~\cite{memgraph_links} extend traditional graph storage engines with GPU-accelerated modules from platforms like NVIDIA cuGraph~\cite{cugraph}.
All these systems support heterogeneous graphs, as such graphs can be trivially represented using labels.
Hence, our approach enables directly extending all these systems with HO.

\textbf{Resource Description Framework (RDF)}~\cite{lassila1998resource} is a
standard to encode knowledge in ontological
models and in RDF stores using triples~\cite{modoni2014survey, harris20094store,
papailiou2012h2rdf}.
There exists a lot of works on RDF and knowledge
graphs~\cite{ristoski2016rdf2vec, portisch2020rdf2vec, jurisch2018rdf2vec,
ristoski2019rdf2vec, lin2015learning, arora2020survey, wang2017knowledge,
dai2020survey, urbani2016kognac, shi2017proje, trouillon2017knowledge,
akrami2020realistic, yao2019kg, sun2019re, chen2020knowledge, lin2015learning,
shi2018open, besta2024hardware}.
While we do not focus on RDFs, both RDF and knowledge graphs also seamlessly enable heterogeneous graphs. Thus, these frameworks could adapt higher-order structures as well.

In addition to GDBs, there are also many \textbf{systems for dynamic, temporal, and streaming graph processing}~\cite{besta2019practice, sakr2021future, choudhury2017nous}.
Their design overlaps with graph databases, because they also focus on
high-performance graph queries and updates, and on solving global analytics.
%
%
Recent efforts work towards processing
subgraphs~\cite{qin2019mining} such as maximal cliques or dense
clusters~\cite{cook2006mining, jiang2013survey, horvath2004cyclic,
chakrabarti2006graph, besta2021sisa, gms, besta2022probgraph}.
Moreover, hybrid frameworks such as AWS GraphStorm~\cite{graphstorm} and Google Spanner Graph~\cite{spannergraph} offer specialized support for GNN training, streaming workloads, and eventually consistent storage.
%
%
These system can even more easily be adapted to use HO as they do not require ACID guarantees.

\textbf{Higher-order graph representations} have gained attention in representation learning, where constructs such as motifs, hyperedges, and subgraphs are leveraged to inject structural bias into GNNs~\cite{morris2020weisfeiler, bunch2020simplicial, yang2022simplicial, goh2022simplicial, giusti2023cell, heydari2022message, feng2019hypergraph, bai2021hypergraph, hajij2020cell, wang2020gognn, bouritsas2022improving, rossi2018hone, qian2022ordered, xu2018powerful, morris2019weisfeiler, hajij2022topological, bodnar2021cellular, bodnar2021topological, arya2020hypersage, yadati2020neural, besta2023hot}. Embedding HO structures (e.g., hyperedges or motifs~\cite{besta2022motif}) into latent vector spaces could further enhance their use in GNNs and retrieval-based graph search~\cite{chen2025interpretable}.
While some systems allow users to extract and process motifs as feature templates or aggregation units, current graph databases lack support for embedding, indexing, or querying HO structures.
Our proposed HO-GDB architecture bridges the gap between higher-order graph structures and learning, and OLTP as well as OLAP workloads in GDBs.

\textbf{Graph pattern matching}~\cite{cook2006mining, jiang2013survey} and in general \textbf{graph mining \& learning}~\cite{besta2021sisa, gms} focus, among others, on efficient searching for various graph patterns, many of which are higher-order structures. Numerous algorithms~\cite{cook2006mining, jiang2013survey, horvath2004cyclic,
chakrabarti2006graph, duan2019dynamic, besta2015accelerating, besta2020highcolor, gianinazzi2021parallel, gianinazzi2021learning, strausz2022asynchronous, besta2017push, bernstein2021deamortization, neiman2015simple, solomonik2017scaling, gianinazzi2018communication, qin2019mining, 
besta2020substream} and frameworks~\cite{cook2006mining, jiang2013survey, horvath2004cyclic, besta2019slim,
chakrabarti2006graph, besta2021sisa, gms, besta2022probgraph} have been developed. These schemes are related to this work in that one could harness the algorithmic solutions to provide more efficient implementations of various OLAP analytics related to mining HO structures within a HO-GDB.

Many \textbf{workload specifications and benchmarks for GDBs} exist,
covering OLTP interactive queries (SNB~\cite{ldbc_snb_specification},
LinkBench~\cite{armstrong2013linkbench}, and BG~\cite{barahmand2013bg}), OLAP
workloads (Graphalytics~\cite{ldbc_graphanalytics_paper}), or business
intelligence queries (BI~\cite{DBLP:journals/pvldb/SzarnyasWSSBWZB22,
early_ldbc_paper}).
%
%
Moreover, there are many evaluations of GDBs~\cite{capotua2015graphalytics,
dominguez2010survey, mccoll2014performance, jouili2013empirical,
ciglan2012benchmarking, lissandrini2017evaluation, lissandrini2018beyond,
tian2019synergistic}.
%
%
These workload specifications could be extended with HO workloads such as HO GNNs~\cite{besta2024demystifying}.

\section{CONCLUSION}

\enlargeSQ

Higher-order (HO) graph workloads such as subgraph-level GNNs, hypergraph traversals, and motif-based analytics are rapidly becoming central to achieving state-of-the-art performance in many graph-based machine learning and data science tasks. Yet, despite the widespread adoption and maturity of graph database systems (GDBs), these systems remain deeply rooted in first-order data models and lack native support for HO  constructs. This disconnect has left a growing set of accuracy-critical applications unsupported.

In this work, we introduce a new class of systems: Higher-Order Graph Databases (HO-GDBs), which extend the capabilities of traditional graph databases to support HO structures as first-class citizens. We present a formal framework for encoding HO constructs—such as hyperedges, node-tuples, and subgraphs—into the Labeled Property Graph (LPG) model using lossless lowering and lifting transformations. We then tackle the key system design challenges posed by this integration: how to ensure ACID-compliant transactions over HO entities, how to define declarative APIs for both OLTP and OLAP workloads, and how to maintain performance through motif-level parallelism. Our prototype, implemented atop Neo4j~\cite{neo4j_book}, demonstrates that HO integration enables substantial gains in the expressiveness and accuracy of OLAP tasks, while preserving transactional semantics and scalability of OLTP queries. In particular, our evaluation shows that higher-order GNNs trained over HO-GDB representations yield significantly better accuracy and faster convergence compared to their non-HO counterparts. HO-GBD could be further extended to other types of databases such as document stores~\cite{besta2023demystifying}, or it could harness various architectural acceleration techniques~\cite{besta2024hardware, di2022building, di2019network, fompi-paper}. As the lowerings of HO graphs have increased numbers of vertices and edges, they could be combined with various forms of structural graph compression~\cite{besta2018log, besta2018survey, besta2019slim} to accelerate their processing through smaller sizes.

Overall, this work introduces a paradigm shift in GDBs -- one that tightly couples HO learning and reasoning with transactional graph storage and querying. As demand for richer modeling and more accurate learning grows across industry and science, HO-GDBs may lead a new generation of graph systems capable of supporting advanced analytics. By extending HO constructs at the level of the core data model -- the LPG model -- our approach remains compatible with widely used architectures. This foundational integration enables seamless adoption across transactional GDBs, RDF systems, knowledge graphs, and eventually consistent graph engines, unifying HO semantics across storage, querying, and learning.

\iftr
\section*{Acknowledgements}

We thank Hussein Harake, Colin McMurtrie, Mark Klein, Angelo Mangili, and the whole CSCS team granting access to the Ault, Daint and Alps machines, and for their excellent technical support.
We thank Timo Schneider for help with infrastructure at SPCL, and Florian Scheidl for help at the early stages of the project.
This project received funding from the European Research Council (Project PSAP, No.~101002047), and the European High-Performance Computing Joint Undertaking (JU) under grant agreement No.~955513 (MAELSTROM). This project was supported by the ETH Future Computing Laboratory (EFCL), financed by a donation from Huawei Technologies. This project received funding from the European Union's HE research and innovation programme under the grant agreement No. 101070141 (Project GLACIATION).
We gratefully acknowledge Polish high-performance computing infrastructure PLGrid (HPC Center: ACK Cyfronet AGH) for providing computer facilities and support within computational grant no.~PLG/2024/017103.
\fi

\bibliographystyle{ACM-Reference-Format}
\bibliography{references.complete}

\iftr
\appendix
\section{Extending the Graph Data Model with Higher-Order}
\label{sec:appendix}

\subsection{Notation}

We summarize the mathematical notation, that is used throughout the paper in Table~\ref{tab:notation}.

\begin{table}[ht!]
    \small
    \centering
    \begin{tabular}{@{}ll@{}}
        \toprule
        \textbf{Symbol} & \textbf{Description} \\
        \midrule
        $H=(V_H,E_H,\ldots)$        & Higher-order graph under discussion. \\
        $n = |V_H|$                 & Number of (plain) vertices. \\
        $m = |E_H|$                 & Number of (plain) edges. \\
        $d$                         & Length of the feature vector stored on every \\
                                    & entity (nodes, edges, HO objects). \\
        $n_h = |h|$                 & Cardinality of a hyperedge $h\in E_H$. \\
        $n_H$                       & Number of hyperedges.\\
        $T$                         & Set of node-tuples in a graph with a node-tuple \\
                                    & collection. \\
        $n_T = |T|$                 & Number of node-tuples. \\
        $t = (v_0,\dots,v_{n_t-1})$ & single node-tuple; $|t|=n_t$ is its arity. \\
        $t_i$                       & $i$-th element of the tuple $t$. \\
        $S$                         & Set of subgraphs in a graph with a subgraph \\
                                    & collection. \\
        $F$                         & Set of subgraph-edges. \\
        $n_S = |S|$                 & Number of subgraphs. \\
        $n_F = |F|$                 & Number of subgraph-edges. \\
        $s=(V_s,E_s)\in S$          & Single subgraph with its own vertex/edge set. \\
        $n_s = |V_s|$               & Number of vertices inside subgraph $s$. \\
        $m_s = |E_s|$               & Number of edges inside subgraph $s$. \\
        \bottomrule
    \end{tabular}
    \caption{Mathematical notation used throughout the paper.}
    \label{tab:notation}
\end{table}

\subsection{Hypergraphs}

\sloppy Given an undirected hypergraph $H=(V_H,E_H,x_H)$ we construct a bipartite heterogeneous graph $G=(V_G,E_G,\tau,\kappa,x_G)$ via the lowering $L^\top_H$:
\begin{align*}
    V_G &= V_H \cup E_H,\\
    E_G &= \{(v,e),(e,v)\mid v\in V_H,\;e\in E_H,\;v\in e\},\\
    \tau_u &=
    \begin{cases}
        \texttt{\_node}      & u\in V_H,\\
        \texttt{\_hyperedge} & u\in E_H,
    \end{cases}\\
    \kappa_{(u,v)} &= \texttt{\_incidence}\qquad\forall(u,v)\in E_G,\\
    x_G(u) &= x_H(u)\qquad\forall u\in V_G.
\end{align*}

We recover $H$ from $G$ via lifting $L_H$:
\begin{align*}
    V_H &= \{v\in V_G\mid\tau_v=\texttt{\_node}\},\\
    E_H &= \bigl\{\{u\mid(u,h)\in E_G\}\;\big|\;
          h\in V_G,\,\tau_h=\texttt{\_hyperedge}\bigr\},\\
    x_H(u) &= x_G(u)\qquad\forall u\in V_H\cup E_H.
\end{align*}

\begin{theorem}
    \label{thm:hgproof}
    $L^\top_H$ and $L_H$ are isomorphism-preserving.
\end{theorem}
\begin{proof}
    Consider two hypergraphs $H_1$ and $H_2$ such that $H_1\cong H_2$. Then there exist bijections $\varphi_V$ and $\varphi_E$ such that for any hyperedge $e = \{v_1, \dots, v_k\} \in E_{H_1}$, $\varphi_E(e) = \{ \varphi_V(v_1), \dots, \varphi_V(v_k) \}$ (incidence is preserved), and the features are preserved: $x_{H_1}(u) = x_{H_2}(\varphi(u))$ for all $u \in V_{H_1} \cup E_{H_1}$. Let $G_1 = L^T_H(H_1)$ and $G_2 = L^T_H(H_2)$. Consider $\psi: V_{G_1} \rightarrow V_{G_2}$ defined as:
    \[
        \psi(v) = 
        \begin{cases}
            \varphi_V(v), & v \in V_{H_1}, \\
            \varphi_E(v), & v \in E_{H_1}.
        \end{cases}
    \]
    This function preserves edges (since $(v, e) \in E_{G_1} \iff (\psi(v), \psi(e)) \in E_{G_2}$), node types (since $\tau_1(v) = \tau_2(\psi(v))$), and features (since $x_{G_1}(v) = x_{G_2}(\psi(v))$). Thus $G_1\cong G_2 \impliedby H_1 \cong H_2$.
    
    Conversely, let $\psi:V_{G_1}\to V_{G_2}$ be a graph isomorphism. Since $\psi$ preserves types $\tau$, $\psi$ must yield independent bijections $\varphi_V=\psi|_{V_{H_1}}$ and $\varphi_E=\psi|_{E_{H_1}}$. Since $\psi$ preserves incidence, $e=\{v_i\}\in E_{H_1}\iff\varphi_E(e)=\{\varphi_V(v_i)\}\in E_{H_2}$; thus $(\varphi_V,\varphi_E)$ is a hypergraph isomorphism, implying $G_1\cong G_2 \implies H_1 \cong H_2$. Thus, $G_1\cong G_2 \iff H_1 \cong H_2$.

    The proof for $L_H$ follows through symmetry since $L_H$ and $L_H^\top$ are mutual inverses.
\end{proof}

The transformation is thus lossless, following theorem~\ref{thm:loss}; the transformation is depicted visually in the top row of Figure~\ref{fig:scheme}. If each hyperedge has size $O(n)$, the lowered heterogeneous graph $G$ has $n+n_H$ vertices and $\Sigma_{h} n_h \in O(nn_H)$ edges.  the memory complexities of the heterogeneous graph is $O(d(nn_H))$, while the lowering and lifting run in $O(d(nn_H))$ time. Graph updates can be performed by these operations:

\begin{enumerate}
    \item \emph{Inserting node $v_0$.} \\
    $V_G \gets V_G\cup \{v_0\}$
    \item \emph{Deleting node $v_0$.} \\
    $H_{v_0}:=\{h\mid (v_0,h)\in E_G,\;\tau_h=\texttt{\_hyperedge},\;|h|=2\}$\\
    $E_G\gets E_G\setminus\bigl\{(v_0,h),(h,v_0)\mid h\in V_G\bigr\}$\\
    $V_G\gets V_G\setminus\bigl(\{v_0\}\cup H_{v_0}\bigr)$
    \item \emph{Inserting hyperedge $e_0$.} \\
    $V_G \gets V_G\cup \{e_0\}$ \\
    $E_G \gets E_G \cup \{(v, e_0), (e_0, v) \mid  v\in e_0\}$
    \item \emph{Deleting hyperedge $e_0$.} \\
    $E_G \gets E_G \setminus \{(v, e_0), (e_0, v) \mid (e_0, v)\in E_G\}$\\
    $V_G \gets V_G\setminus \{e_0\}$ 
\end{enumerate}

\subsection{Simplicial Complexes}
Simplicial complexes are lowered exactly as hypergraphs. Because every subset of a simplex is itself a simplex, we may
alternatively replace each maximal simplex by the complete set of pairwise edges. This clique-lowering is lossless for simplicial complexes but not for arbitrary hypergraphs.

\subsection{Graphs with Node-Tuple Collections}

\sloppy Given an undirected graph with a node-tuple collection $H=(V_H, E_H, T, x_H)$, we can construct a tripartite heterogeneous multigraph $G=(V_G, E_G, \tau, \kappa, x_G)$ using the lowering $L_N^\top$:
\begin{align*}
    V_G &= V_H \cup E_H \cup T,\\
    E_G &= \{(v,e),(e,v)\mid e=(v,\cdot)\text{ or }e=(\cdot,v)\}\\
      &\cup\{(v,t)\mid t\in T,\;t_i=v\},\\
    \tau_u &=
    \begin{cases}
        \texttt{\_node}       & u\in V_H,\\
        \texttt{\_edge}       & u\in E_H,\\
        \texttt{\_node\_tuple}& u\in T,
    \end{cases}\\
    \kappa_{(u,v)} &=
    \begin{cases}
        \texttt{\_adjacency} & \tau_u=\texttt{\_node},\;\tau_v=\texttt{\_edge},\\
        \texttt{\_adjacency} & \tau_u=\texttt{\_edge},\;\tau_v=\texttt{\_node},\\
        (\texttt{\_node\_membership},i) & \tau_v=\texttt{\_node\_tuple},\;t_i=u,\\
    \end{cases}\\
    x_G(u) &= x_H(u)\quad \forall\ u\in V_G.
\end{align*}

\sloppy We recover $H$ from $G$ via lifting $L_N$:
\begin{align*}
    V_H &= \{v\in V_G\mid\tau_v=\texttt{\_node}\},\\
    E_H &= \bigl\{(u,v)\mid u,v\in V_H,\;
          \exists e\!:\!\tau_e=\texttt{\_edge},\,
          (u,e),(e,v)\!\in E_G\bigr\},\\
    T   &= \bigl\{(u_0,\dots,u_{k-1})\mid
          \exists t\!:\!\tau_t=\texttt{\_node\_tuple},\\
        &\hspace{5em}(u_i,t)\in E_G,\;
          \kappa_{(u_i,t)}=(\texttt{\_node\_membership},i)\bigr\},\\
    x_H(u) &= x_G(u)\qquad\forall u\in V_H\cup E_H\cup T_H.
\end{align*}

\begin{theorem}
    $L^\top_N$ and $L_N$ are isomorphism-preserving.
\end{theorem}
\begin{proof}
  Similar to the proof of theorem~\ref{thm:hgproof}: we replace the two vertex classes
  \texttt{\_node},\,\texttt{\_hyperedge} by the three classes
  \texttt{\_node},\,\texttt{\_edge},\,\texttt{\_node\_tuple},
  and replace the incidence tag \texttt{\_incidence} by the tags \texttt{\_adjacency} and $(\texttt{\_node\_membership},i)$.
\end{proof}

\sloppy The transformation is shown visually in the middle row of Figure~\ref{fig:scheme}. The lowered heterogeneous graph $G$ has $n+m+n_T$ vertices and $2m+\Sigma_{t} n_t \in O(m+nn_T)$  edges. Thus including $d$ features for each entity, the time complexity of the lowering and lifting is $O(d(m+nn_T))$.

\begin{enumerate}
    \item \emph{Insert node} $v_0$: \\
    $V_G \gets V_G \cup \{v_0\}$
    \item \emph{Delete node} $v_0$: \\
    $H_{v_0} \gets \{e \mid (v_0, e)\in E_G, \tau_e = \texttt{\_edge}\}$ \\
    $T_{v_0} \gets \{ (v_0, t) \mid (v_0, t) \in E_G, \tau_t=\texttt{\_node\_tuple} \}$\\
    $E_G \gets E_G \setminus (\{ (v_0, e), (e, v_0) \mid e \in H_{v_0} \} \cup T_{v_0})$ \\
    $V_G \gets V_G \setminus (\{v_0\} \cup H_{v_0})$
    \item \emph{Insert edge} $e_0$: \\
    $V_G \gets V_G \cup \{e_0\}$ \\
    $E_G \gets E_G \cup \{ (v, e_0), (e_0, v) \mid v \in e_0 \}$ 
    \item \emph{Delete edge} $e_0$: \\
    $E_G \gets E_G \setminus \{ (v, e_0), (e_0, v) \mid (e_0, v) \in E_G \}$ \\
    $V_G \gets V_G \setminus \{e_0\}$
    \item \emph{Insert tuple} $t_0$: \\
    $V_G \gets V_G \cup \{t_0\}$ \\
    $E_G \gets E_G \cup \{ (v, t_0) \mid t_{0i} = v,\, 0 \le i < n_{t_0} \}$
    \item \emph{Delete tuple} $t_0$: \\
    $E_G \gets E_G \setminus \{ (v, t_0) \mid (v, t_0) \in E_G \}$ \\
    $V_G \gets V_G \setminus \{t_0\}$ 
\end{enumerate}

\subsection{Graphs with Subgraph Collections}

For an undirected graph with a subgraph collection $H=(V_H, E_H, S, F, x_H)$, we construct a 4-partite heterogeneous graph $G=(V_G, E_G, \tau, \kappa, x_G)$ by the lowering $L^\top_S$:
\begin{align*}
    V_G &= V_H \cup E_H \cup S \cup F, \\
    E_G &= \{ (c, d), (d, c) \mid c \in d,\, c \in V_H,\, d \in E_H \} \\
    &\cup \{ (c, s) \mid c \in V_s,\, c \in V_H,\, s \in S \} \\
    &\cup \{ (e, s) \mid e \in E_s,\, e \in E_H,\, s \in S \} \\
    &\cup \{ (s, f), (f, s) \mid s \in f,\, s \in S,\, f \in F \}, \\
    \tau_u &=
    \begin{cases}
        \texttt{\_node} &\text{if } u \in V_H, \\
        \texttt{\_edge} &\text{if } u \in E_H, \\
        \texttt{\_subgraph} &\text{if } u \in S, \\
        \texttt{\_subgraph\_edge} &\text{if } u \in F,
    \end{cases} \\
    \kappa_{(u,v)} &=
    \begin{cases}
        \texttt{\_adjacency} &\text{if } (u,v) \in V_H \times E_H, \\
        \texttt{\_adjacency} &\text{if } (v,u) \in V_H \times E_H, \\
        \texttt{\_node\_membership} &\text{if } (u,v) \in V_H \times S, \\
        \texttt{\_edge\_membership} &\text{if } (u,v) \in E_H \times S, \\
        \texttt{\_subgraph\_adjacency} &\text{if } (u,v) \in S \times F, \\
        \texttt{\_subgraph\_adjacency} &\text{if } (v,u) \in F \times S,
    \end{cases} \\
    x_G(u) &= x_H(u) \quad \forall u \in V_G.
\end{align*}

We recover $H$ from $G$ via lifting $L_S$:
\begin{align*}
    V_H &= \{ v \in V_G \mid \tau_v = \texttt{\_node} \}, \\
    E_H &= \{ (u, v) \mid u,v \in V_H,\, \exists e \in V_G,\, \tau_e = \texttt{\_edge}, \\
    &\hspace{3em} (u,e) \in E_G,\, (e,v) \in E_G \}, \\
    S &= \Big\{ \big( \{ v \in V_H \mid (v,s) \in E_G \},\, \\
    &\hspace{3em} \{ e \in E_H \mid (e,s) \in E_G \} \big) \mid s\in V_G, \tau_s = \texttt{\_subgraph} \Big\}, \\
    F &= \{ (u,v) \mid u,v \in S,\, \exists e \in V_G,\, \tau_e = \texttt{\_subgraph\_edge}, \\
    &\hspace{3em} (u,e) \in E_G,\, (e,v) \in E_G \}, \\
    x_H(u) &= x_G(u) \quad \forall u \in V_H \cup E_H \cup S \cup F.
\end{align*}

\begin{theorem}
    $L^\top_S$ and $L_S$ are isomorphism-preserving.
\end{theorem}

\sloppy The proof is similar to theorem~\ref{thm:hgproof}. The transformation is depicted visually in the bottom row of Figure~\ref{fig:scheme}. The lowered heterogeneous graph $G$ has $n+m+n_S+n_F$ vertices and $2m+\Sigma_{s} (n_s+m_s) \in O(m+nn_S+mn_S+n_F)$ edges. Thus including $d$ features for each entity, the time complexity of the lowering and lifting are each $O(d(n_F+nn_S+mn_S))$. Graph updates can be performed by the following operations:

\begin{enumerate}
    \item \emph{Insert node} $v_0$: \\
    $V_G \gets V_G \cup \{v_0\}$
    \item \emph{Delete node} $v_0$: \\
    $H_{v_0} \gets \{e \mid (v_0, e)\in E_G, \tau_e = \texttt{\_edge}\}$ \\
    $S_{v_0} \gets \{ (v, s) \mid v\in \{v_0\}\cup H_{v_0}, (v, s) \in E_G, \tau_s=\texttt{\_subgraph} \}$\\
    $E_G \gets E_G \setminus (\{ (v,e), (e,v) \mid e \in H_{v_0}, v\in V_G \} \cup S_{v_0})$ \\
    $V_G \gets V_G \setminus (\{v_0\} \cup H_{v_0})$
    \item \emph{Insert edge} $e_0$: \\
    $V_G \gets V_G \cup \{e_0\}$ \\
    $E_G \gets E_G \cup \{ (v,e_0), (e_0,v) \mid v \in e_0 \}$
    \item \emph{Delete edge} $e_0$: \\
    $S_{e_0} \gets \{ (e_0, s) \mid (e_0, s) \in E_G, \tau_s=\texttt{\_subgraph} \}$\\
    $E_G \gets E_G \setminus (\{ (v,e_0), (e_0,v) \mid (e_0,v) \in E_G \}\cup S_{e_0})$ \\
    $V_G \gets V_G \setminus \{e_0\}$
    \item \emph{Insert subgraph} $s_0$: \\
    $V_G \gets V_G \cup E_s \cup \{s_0\}$ \\
    $E_G \gets E_G \cup \{ (v,s_0) \mid v \in V_{s_0} \} \cup \{ (e,s_0) \mid e \in E_{s_0} \}$
    \item \emph{Delete subgraph} $s_0$: \\
    $H_{s_0} \gets \{f \mid (s_0, f)\in E_G, \tau_f = \texttt{\_subgraph\_edge}\}$ \\
    $S_{s_0}\gets \{(v, s_0)\mid (v, s_0)\in E_G, \tau_v\in\{\texttt{\_node}, \texttt{\_edge}\}\}$\\
    $E_G \gets E_G \setminus (\{ (f,s), (s, f) \mid f \in H_{s_0}, s\in V_G \} \cup S_{s_0})$ \\
    $V_G \gets V_G \setminus (\{s_0\} \cup H_{s_0})$ 
    \item \emph{Insert subgraph-edge} $f_0$: \\
    $V_G \gets V_G \cup \{f_0\}$ \\
    $E_G \gets E_G \cup \{ (s,f_0), (f_0,s) \mid s \in f_0\}$
    \item \emph{Delete subgraph-edge} $f_0$: \\
    $E_G \gets E_G \setminus \{ (s,f_0), (f_0,s) \mid (s,f_0) \in E_G \}$ \\
    $V_G \gets V_G \setminus \{f_0\}$ 
\end{enumerate}

\fi

\end{document}